\documentclass[11pt]{amsart}
\usepackage[latin1]{inputenc}
\usepackage{graphicx}

\def\RR{\mathbb{R}}

\def\PP{\mathbb{P}}

\def\EE{\mathbb{E}}

\newcommand*{\wh}{\widehat}
\newcommand*{\wt}{\widetilde}

\usepackage{amsmath,amssymb,latexsym,amsthm,amsfonts,enumerate}
\usepackage{verbatim}
\usepackage{enumerate}
\usepackage{hyperref}

\theoremstyle{plain}
\newtheorem{theorem}{Theorem}
\newtheorem{proposition}[theorem]{Proposition}
\newtheorem{corollary}[theorem]{Corollary}
\newtheorem{lemma}[theorem]{Lemma}

\theoremstyle{remark}

\newtheorem*{remarks}{Remarks}

\usepackage{color}
    \definecolor{Red}{rgb}{1.00, 0.00, 0.00}
    
    \definecolor{DRed}{rgb}{0.5, 0.00, 0.00}
    
    \definecolor{Blue}{rgb}{0.00, 0.00, 1.00}
    
    \definecolor{Green}{rgb}{0.3, 0.7, 0.3}
    
    \definecolor{PaleGrey}{rgb}{.6, .6, .6}

\begin{document}

\title{A new look at short-term implied volatility in asset
    price models with jumps}

\author{Aleksandar Mijatovi\'{c}}
\address{Department of Mathematics, Imperial College London, UK}
\email{a.mijatovic@imperial.ac.uk}

\author{Peter Tankov} 
\address{LPMA, Universit\'e Paris-Diderot}
\email{peter.tankov@polytechnique.org}

\begin{abstract}
We analyse the behaviour of the implied volatility smile for options
close to expiry in the exponential L\'evy class of asset price models 
with jumps.  We introduce a new renormalisation of the strike variable 
with the property that the implied volatility 
converges to a non-constant limiting shape, 
which is a function of 
both the diffusion component of the process and the jump activity (Blumenthal-Getoor) 
index of the jump  component. Our limiting implied volatility formula relates
the jump activity of the underlying asset price process 
to the short end of the implied volatility surface 
and 
sheds new light on the 
difference between finite and infinite variation jumps from the 
viewpoint of option prices:  in the latter, 
the wings of the limiting smile are determined by the jump activity indices of
the positive and negative jumps, whereas in the former, 
the wings have a constant model-independent slope. This 
result gives a theoretical justification for the preference 
of the infinite variation L\'evy models 
over the finite variation ones in 
the calibration based on short-maturity option prices. 
\end{abstract}



\maketitle

\section{Introduction}

In financial markets, the price of a vanilla call or put option 
on a risky asset 
with strike 
$e^k$
and maturity 
$t$
is often quoted in terms of the \emph{implied volatility}
$\wh\sigma(t,k)$
(see~\eqref{eq:DefOfImpiedVol}
in Section~\ref{sec:Asym_IVOL}
for the definition
and~\cite{gatheral} for more information on implied
volatility).
Similarly, given a risk-neutral pricing model,
one can define a function 
$(t,k)\mapsto\wh\sigma(t,k)$
via the prices of the 
vanilla options under that model. 
The implied volatility is a central object in 
option markets and it is
therefore not surprising that understanding the properties 
and computing the function
$(t,k)\mapsto\wh\sigma(t,k)$
for widely used  pricing models 
has been of considerable interest in the mathematical 
finance literature.
Typically, for a given modelling framework, 
the implied volatility
$\wh\sigma(t,k)$
is not available in closed form.
Hence the study of the asymptotic behaviour
in a variety of asymptotic regimes (e.g. fixed
$t$
and
$k\to\pm\infty$~\cite{lee.04,benaim.friz.06,gulisashvili.10};
$t\to\infty$
with 
$k$ constant~\cite{tehranchi.09} or proportional~\cite{KJM} to 
$t$;
$t\to0$
and 
$k$
constant~\cite{Roper_Rutkowski_09,tankov.09,Ford_Figueroa-Lopez}
etc.)
has attracted a lot of attention in the recent years.



In this paper 
we 
assume that 
the returns of the 
risky asset 
$S = e^{X}$
are modelled by 
a L\'evy process $X$ 
and 
study the relationship
between the jump activity of $X$ 
and the implied volatility 
at short maturities
in the model $S$.
Most existing approaches analyse either the at-the-money case, when
the implied volatility is determined exclusively by the diffusion
component and converges to zero in the pure jump models
(see~\cite[Prop.~5]{tankov.09}, \cite{muhle2011small,gong.al.11}), or the 
fixed-strike out-of-the-money case, when the implied volatility 
for short maturities 
explodes in the presence of jumps
(\cite{Roper},~\cite{Ford_Figueroa-Lopez},~\cite{tankov.09}). 
However, in the option markets, 
(a) although the implied volatility for liquid strikes grows with decreasing
$t$,
it remains within a range
of reasonable values and appears not to explode, 
and (b) 
the liquid strikes become concentrated around the money as 
the maturity gets shorter. For instance, in the FX option markets, 
which are among the most liquid derivatives markets in the world,
options with fixed values of the Black-Scholes delta are quoted for 
each maturity (see~\cite{AL12} for the details on the conventions in
FX option markets and a natural parameterisation of the smile using
the Black-Scholes delta).
\begin{figure}
\label{fig:Strikes}
\begin{center}
\includegraphics[width = \textwidth]{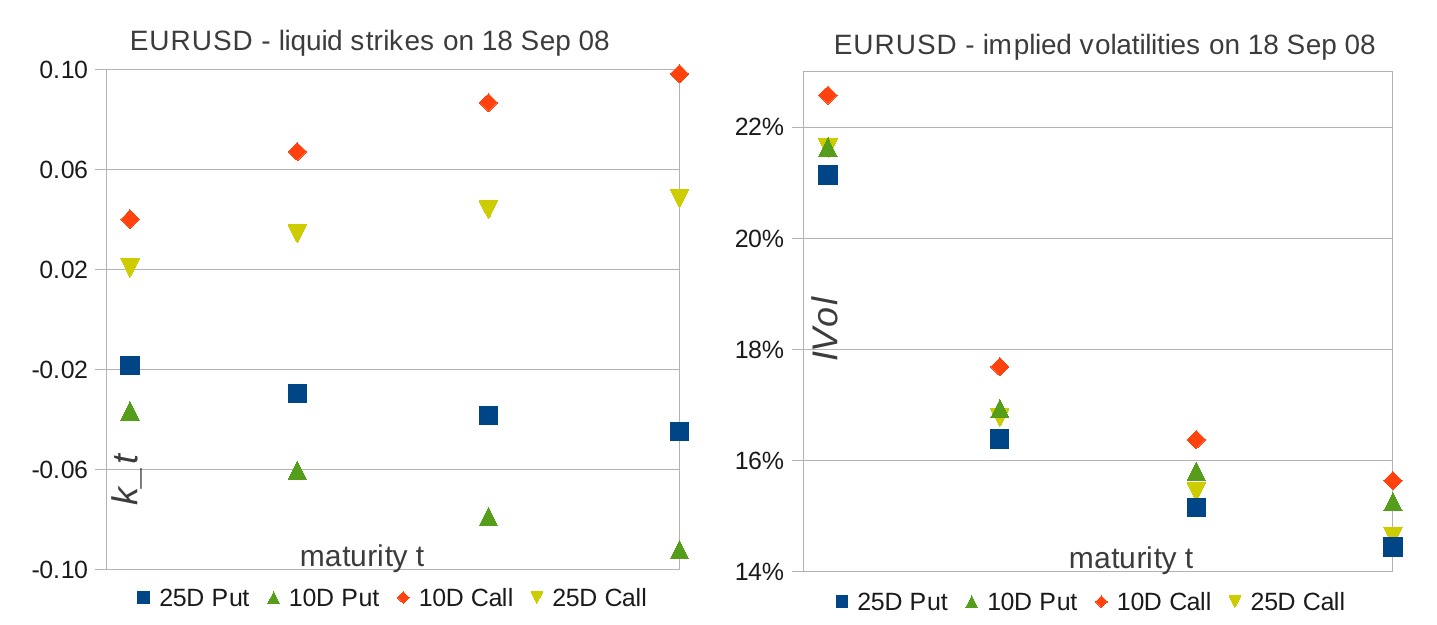}
\caption{\footnotesize The liquid 10-$\Delta$  and 25-$\Delta$ strikes 
(left panel) and the corresponding implied volatilities (right panel)
for market defined maturities $t\in\{\text{$1$ day, $1$ week, $2$ weeks, $1$ month}\}$
in USDJPY suggest the following: as maturity $t$ becomes
small, the relevant strikes 
$k_t$ should approach the at-the-money strike and the implied volatilities
$\wh \sigma(t,k_t)$
should remain bounded.}  
\end{center}
\end{figure}
The market data in Figure~\ref{fig:Strikes} 
therefore suggests that, in order 
to understand the behaviour of the volatility surface at 
short maturities,
one should look for a \textit{moving log-strike}
$k_t\neq0$,
for 
$t>0$,
such that (i)~the corresponding implied 
volatility has a non-trivial limit
$\lim_{t\downarrow0}\wh \sigma(t,k_t)$
and 
(ii)~the strike 
$k_t$
converges to the at-the-money strike
as maturity 
$t$
tends to zero
(i.e. $\lim_{t\downarrow0}k_t=0$).


This paper defines a new  
universal and model-free parameterisation of the log-strike
given by
$$
k_t=\theta \sqrt{t \log(1/t)}\quad \text{where}\quad
\theta\in\RR\!\setminus\!\{0\}.
$$
For fixed $\theta$, the corresponding strike value
tends to the at-the-money strike as
$t\downarrow0$
but is out-of-the-money
for each short maturity
$t>0$. We prove that under suitable assumptions the limiting implied volatility 
$\sigma_0(\theta)=\lim_{t\downarrow0} \wh \sigma(t,k_t)$
takes the following
form as a function of 
$\theta$:
\begin{equation}
\label{eq:The_Formula}
\sigma_0(\theta)
= \max\left\{\frac{-\theta}{\sqrt{1-(\alpha_--1)^+}},
\sigma,\frac{\theta}{\sqrt{1-(\alpha_+-1)^+}}\right\}
\qquad
\text{for any}\qquad 
\theta\in\RR\!\setminus\!\{0\}.
\end{equation}
In this formula
$\sigma$
denotes the volatility of the Gaussian component
of the underlying L\'evy process 
$X$
and 
$\alpha_+$
(resp. 
$\alpha_-$)
denotes the jump activity (Blumenthal-Getoor) index of the positive
(resp. negative) jumps of $X$. More precisely, if the jump
measure of $X$ is denoted by $\nu$, $\alpha_+$ and $\alpha_-$ are given
by
$$
\alpha_+ = \inf\{p\geq0:\int_{(0,1)}|x|^p\nu(dx)<\infty\}\quad
\text{and}\quad \alpha_- = \inf\{p\geq0:\int_{(-1,0)}|x|^p\nu(dx)<\infty\}.
$$

Unlike in the case of fixed strike, where 
short maturity smile explodes in the presence of jumps, 
our parameterisation of the strike as a function of time 
yields a non-constant formula for the limiting implied volatility, 
which depends on the balance between the size of the Gaussian volatility parameter
and the activity of small jumps.
It allows us to make the following observations about the relationship between
the short-dated option prices and the characteristics of the underlying
model:
\begin{enumerate}[(i)]
\item the formula for
$\sigma_0(\theta)$
depends on the jump measure of the log-spot process
$X$
only if the jumps are of infinite variation;
put differently, if the jumps of 
$X$
are of finite variation, then the absolute value of the slope of the limiting smile 
for large
$|\theta|$
is equal to one 
and in particular
$\sigma_0(\theta)$
does not depend on the structure of jumps;
\label{rem:Theoretical}
\item the limiting smile
$\sigma_0(\theta)$
is V-shaped in the absence of the diffusion component
(i.e. when 
$\sigma=0$)
and is 
U-shaped otherwise;
\end{enumerate}

Remark~\eqref{rem:Theoretical} 
provides a theoretical basis
for distinguishing between the models with jumps of finite and infinite variation 
in terms of the observed prices of vanilla options with short maturity. 
It is well-known that, 
for any short maturity $t$, 
the market implied smile 
$k\mapsto\wh\sigma(t,k)$
exhibits pronounced skewness
and/or curvature,
due, in particular, to the risk of large moves over short time horizons perceived by the investors.
Hence, jumps are typically introduced into the risk-neutral pricing models with the 
aim to capture this risk and modulate the at-the-money skew 
of the implied volatility 
$\wh \sigma(t,k)$
at small $t$ 
(see e.g.~\cite[Eq~(5.10)]{gatheral}).
However, since this task can be accomplished by jumps of either finite
or infinite variation, this requirement tells us little about the options
implied jump activity of the underlying risk-neutral model.
On the other hand,
the formula for 
$\sigma_0(\theta)$
implies that, if we need to control the tails (in
the parameter
$\theta$) 
of the implied volatility for short maturities, we must
use jumps of infinite variation. This finding complements
the analysis in~\cite{carrwu}
of the path-wise structure of the risk-neutral process implied by the 
option prices on the S\&P 500 index. 


In recent years, there has been a lot of interest in the literature on the statistics
of stochastic process in the question of the estimation of the Blumenthal-Getoor
index of models with jumps based on high-frequency data.
For example,
it is shown in~\cite{aj4}
that the jump activity (measured by the Blumenthal-Getoor index) estimated on
high-frequency stock returns for two large US corporates is well beyond one, 
implying that the underlying model for stock returns should have jumps of 
infinite variation.
Likewise, the formula in~\eqref{eq:The_Formula} suggests that jumps of 
infinite variation are needed in order to capture the correct tails 
(in $\theta$)
of the quoted short-dated option prices. 

The formula in~\eqref{eq:The_Formula}
follows from Corollary~\ref{volimp.cor},
which gives the expansion of the implied volatility 
$\wh \sigma(t,k_t)$,
where
$k_t=\theta \sqrt{t \log(1/t)}$, 
up to order
$o\left(1/\log(1/t)\right)$.
This expansion is consequence of 
(A) Theorem~\ref{ivexpand},
which itself gives an expansion of the implied volatility 
for a general log-strike
$k_t$
that tends
to zero as
$t\downarrow0$,
and (B) 
Theorem~\ref{infvar}
and Proposition~\ref{finvar},
which describe the asymptotic behaviour of the 
option prices under L\'evy processes with infinite and 
finite jump variations respectively.
Theorem~\ref{ivexpand}
relates the asymptotic behaviour of the vanilla option prices
under a general semimartingale model to the asymptotic behaviour 
of the implied volatility as the log-strike
$k_t$
tends to zero
(it should be noted that the asymptotic regime 
$(t,k_t)$
in Theorem~\ref{ivexpand}
is not covered by the analysis in~\cite{gao.lee.11},
see Remark~\eqref{rem:Not_in_Gao_Lee} after  
Theorem~\ref{ivexpand} for more details).
The asymptotic formula 
in Corollary~\ref{volimp.cor}
then follows by combining 
Theorem~\ref{ivexpand}
with the asymptotic behaviour of the vanilla option prices
established in Theorem~\ref{infvar}
(for the case of jumps of infinite variation)
and Proposition~\ref{finvar}
(for jumps of finite variation).

In a certain sense, 
Theorem~\ref{infvar}
and Proposition~\ref{finvar}
represent the main contributions 
of this paper. The asymptotic formulae for the 
call and put options, 
struck at
$e^{k_t}$
and 
$e^{-k_t}$
respectively,
have the same structure in both results: 
the leading order term is a sum of two contributions,
one coming from the diffusion component of the process
and the other from the jump measure. Which of the two
summands dominates in the limit depends on 
the level of the
parameter
$\theta$.
This structure of the asymptotic formulae is also reflected in the expression
for
$\sigma_0(\theta)$,
as it is clear from~\eqref{eq:The_Formula}
that 
$\sigma_0(\theta)\equiv\sigma$
if 
$\theta$
is between
$-\sigma\sqrt{1-(\alpha_--1)^+}$
and 
$\sigma\sqrt{1-(\alpha_+-1)^+}$, and 
$\sigma_0(\theta)$
only depends on the jump measure otherwise.
However, the proofs of
Theorem~\ref{infvar}
and Proposition~\ref{finvar} 
differ greatly: the finite variation case follows  
from the 
It\^o-Tanaka formula, which can in this case be applied directly
to the hockey-stick payoff function, while the case of jumps with 
infinite variation requires a detailed analysis of the asymptotic behaviour
of the option prices. 

The remainder of the paper is organised as follows. Section~\ref{sec:Option_Prices_Exp_LEvy}
defines the setting and states 
Theorem~\ref{infvar}
and Proposition~\ref{finvar}.
In Section~\ref{sec:Asym_IVOL},
we state and prove the asymptotic formulae for the implied
volatility and establish the limit in~\eqref{eq:The_Formula}.
Section~\ref{sec:Numerical} presents numerical results that demonstrate  
the convergence of option prices and implied volatilities given 
in the previous two sections, in the context of a CGMY model
and a CGMY model with an additional diffusion component. 
Section~\ref{sec:Proofs}
concludes the paper by proving
Theorem~\ref{infvar}
and Proposition~\ref{finvar}.
The appendix contains a short technical lemma, which is applied in
Section~\ref{sec:Proofs}.

\section{Option price asymptotics close to the money}
\label{sec:Option_Prices_Exp_LEvy}

In this paper we study the behaviour of option prices 
close to maturity in an exponential L\'evy model
$S=e^X$,
where
$X$ 
is a L\'evy process 
with the characteristic
triplet 
$(\sigma^2,\nu,\gamma)$.
Throughout the paper
we assume the following:
\begin{enumerate}[$\bullet$]
\item 
$S$
is a true martingale
(i.e. the interest rates and dividend yields are equal);
\item $S$ is normalised to start at 
$S_0=1$
(i.e. as usual
the L\'evy process 
$X$
starts at
$X_0=0$);
\item the tails of the L\'evy measure 
$\nu$
admit exponential moments:
\begin{eqnarray}
\int_{|z|>1} e^{|z|(1+\delta)}
\nu(dz)<\infty\qquad\text{for some $\delta>0$.}
\label{eq:Assum_tail}
\end{eqnarray}
\end{enumerate}
In particular, assumption~\eqref{eq:Assum_tail}
guarantees the finiteness 
of vanilla option prices for any maturity
$t>0$.
Section~\ref{subsec:Infinite_V}
describes the asymptotic behaviour of option prices for short maturities 
in the case the process
$X$
has jumps of infinite variation.
Section~\ref{subsec:Finite_V}
deals with the case where the pure-jump part of 
$X$
has finite variation.

\subsection{L\'evy processes with jumps of infinite variation}
\label{subsec:Infinite_V}
Theorem~\ref{infvar}
describes the asymptotic behaviour of option prices 
in the case the tails of the L\'evy measure of
$X$
around zero have asymptotic power-like behaviour. 
This assumption does not exclude any exponential L\'evy 
models that appear in the literature but yields sufficient 
analytical tractability to characterise a non-trivial 
limit as maturity tends to zero for the option prices 
around the at-the-money.
Before stating the theorem, we recall standard notation 
used throughout the paper: functions
$f(t)$
and
$g(t)$,
where
$g(t)>0$
for all small
$t>0$,
satisfy
\begin{subequations}
\begin{eqnarray}
\label{def:sim}
f(t)\sim g(t) 
\quad
\text{as $t\downarrow 0$ if}
& &
\lim_{t\downarrow 0}\frac{f(t)}{g(t)}=1,\\
\label{def:o}
f(t)=o(g(t))
\quad
\text{as $t\downarrow 0$ if}
& &
\lim_{t\downarrow 0}\frac{f(t)}{g(t)}=0,\\
\label{def:O}
f(t)=O(g(t))
\quad
\text{as $t\downarrow 0$ if}
& &
\text{$\frac{f(t)}{g(t)}$ is bounded
for all small $t>0$.}
\end{eqnarray}
\end{subequations}
Furthermore we denote
$x^+:=\max\{x,0\}$
for any 
$x\in\RR$.

\begin{theorem}\label{infvar}
Let $X$ be a L\'evy process 
as described at the beginning of the section 
and assume that the following holds
\begin{align}
\lim_{x\downarrow 0} x^{\alpha_+} \nu((x,\infty)) = c_+,\qquad
\lim_{x\downarrow 0} x^{\alpha_-} \nu((-\infty,-x)) = c_-
\label{halpha}
\end{align}
for $\alpha_+,\alpha_- \in (1,2)$ and $c_+,c_-\in[0,\infty)$. Let
$k_t$ be a deterministic function satisfying 
$$
k_t>0\quad\forall t>0,\qquad \lim_{t\downarrow 0} k_t = 0
$$
and 
\begin{align*}
&\text{if $\sigma^2=0$}, \quad \lim_{t\downarrow 0} \frac{t^{1/\alpha}}{k_t} = 0\quad 
\text{for some $\alpha\in(\max(\alpha_-,\alpha_+),2)$,}\\
&\text{if $\sigma^2>0$}, \quad \lim_{t\downarrow 0}
\frac{\sqrt{t}}{k_t} = 0. 
\end{align*}
Then, if $c_+>0$, we have
\begin{align}
\mathbb E[(e^{X_t} - e^{k_t})^+] \sim 
\mathbb E[(e^{\sigma W_t- \frac{\sigma^2 t}{2}} -e^{k_t})^+] +  
\frac{t k_t^{1-\alpha_+} c_+}{\alpha_+-1} \quad \text{as $t\downarrow 0$}\label{infvar.eq}
\end{align}
and, if $c_->0$, it holds
\begin{align}
\mathbb E[(e^{-k_t}-e^{X_t})^+] \sim \mathbb E[(e^{-k_t} - e^{\sigma W_t- \frac{\sigma^2 t}{2}})^+] +  \frac{t k_t^{1-\alpha_-}
  c_-}{\alpha_--1} \quad \text{as $t\downarrow 0$}.\label{infvar.eq_Put}
\end{align}
\end{theorem}

\begin{remarks}
\begin{enumerate}[(i)]
\item Theorem~\ref{infvar}
implies that
the price of a call (resp. put) option struck at 
$e^{k_t}$
(resp.
$e^{-k_t}$)
tends to zero at a rate strictly slower than 
$t$
if the paths of the pure jump part of $X$
have infinite variation. 
In particular, combining the notation in~\eqref{def:sim} and~\eqref{def:o},
we get that the following equalities hold as $t\downarrow 0$:
\begin{eqnarray*}
\mathbb E[(e^{X_t} - e^{k_t})^+]  & = &  
\mathbb E[(e^{\sigma W_t- \frac{\sigma^2 t}{2}} -e^{k_t})^+] +  
\frac{t k_t^{1-\alpha_+} c_+}{\alpha_+-1}+ o\left(t k_t^{1-\alpha_+}\right),\\
\mathbb E[(e^{-k_t}-e^{X_t})^+] & = &
\mathbb E[(e^{-k_t} - e^{\sigma W_t- \frac{\sigma^2 t}{2}})^+] +  
\frac{t k_t^{1-\alpha_-}c_-}{\alpha_--1} +
o\left(t k_t^{1-\alpha_-}\right).
\end{eqnarray*}
\label{rem:slower_than_t}
\item The proof of Theorem~\ref{infvar}
is given in Section~\ref{sec:Proof_Thm_infvar}.
\end{enumerate}
\end{remarks}

\subsubsection{Blumenthal-Getoor index and the short-dated option prices.}
\label{subsubsec:BG_index}
Recall that for any L\'evy process
$Y$
with a non-trivial L\'evy measure
$\nu_Y$,
the \textit{Blumenthal-Getoor index}, introduced in~\cite{bg61},
is defined as
\begin{equation}
\label{eq:BG_Def}
\textrm{BG}(Y):=\inf\left\{r\geq0\>:\>\int_{(-1,1)\setminus\{0\}}|x|^r\nu_Y(dx)<\infty\right\}.
\end{equation}
The Blumenthal-Getoor index
is a measure of the jump activity
of the L\'evy process
$Y$, 
since the following holds: 
$r> \textrm{BG}(Y)$
if and only if
$\sum_{s\leq t} |\Delta Y_s|^r<\infty$
almost surely, where 
$\Delta Y_s:=Y_s-Y_{s-}$
denotes the size of the jump of 
$Y$
at time 
$s$.
Furthermore, it is clear from~\eqref{eq:BG_Def}
that 
$\textrm{BG}(Y)$
lies in the interval 
$[0,2]$.

In recent years there has been renewed interest in the
Blumenthal-Getoor index from the point of view of estimation
of the jump structure of stochastic processes based
on high-frequency financial data. For example, it was estimated 
in~\cite{aj4}
that the value of 
$\textrm{BG}(Y)$
is around 
$1.7$
(i.e. the stock price process has jumps of infinite variation)
based on high-frequency transactions
(taken at 
$5$
and
$15$
time intervals) for 
\textit{Intel}
and
\textit{Microsoft}
stocks throughout 2006.

Let
$X^+$
and 
$X^-$
be the pure-jump parts of 
the L\'evy process
$X$
from Theorem~\ref{infvar}.
In other words
$X^+$
(resp. 
$X^-$)
is a L\'evy process
with the characteristic triplet
$(0,\nu^+,0)$
(resp.
$(0,\nu^-,0)$),
where
$\nu^+(dx):=1_{\{x>0\}}\nu(dx)$
(resp.
$\nu^-(dx):=1_{\{x<0\}}\nu(dx)$).
Then assumption~\eqref{halpha}
implies
$$
\textrm{BG}(X^+)=\alpha_+
\qquad\text{and}\qquad
\textrm{BG}(X^-)=\alpha_-,
$$
and relations~\eqref{infvar.eq} and~\eqref{infvar.eq_Put}
of Theorem~\ref{infvar}
describe how the Blumenthal-Getoor indices of 
the positive and negative jumps of 
$X$
influence the asymptotic behaviour of 
option prices at short maturities. 
The result clearly depends on the 
asymptotic behaviour of the log-strike
$k_t$.
In Section~\ref{sec:Asym_IVOL}
we will prescribe a specific parametric form of
$k_t$
(see~\eqref{eq:def_kt})
and give explicit formulae for 
the asymptotic expansion and the limit 
of the implied volatility as maturity tends
to zero in terms of 
the Blumenthal-Getoor indices  of
$X^+$
and
$X^-$
(see Corollary~\ref{volimp.cor}
for details).

\subsection{L\'evy processes with jumps of finite variation}
\label{subsec:Finite_V}
In this section we study the option price asymptotics 
at short maturities
in the case the process 
$X$
has a (possibly trivial) Brownian component 
and a pure jump part of finite variation.

\begin{proposition}\label{finvar}
Let $X$ be a L\'evy process 
as described at the beginning of Section~\ref{sec:Option_Prices_Exp_LEvy}.
Assume further that the jump part of 
$X$
has finite variation, i.e.
\begin{align*}
\int_{\mathbb R\setminus\{0\}} |x|\nu(dx) <\infty.
\end{align*}
Let
$k_t$ be a deterministic function satisfying 
$$
k_t>0\quad\forall t>0,\qquad \lim_{t\downarrow 0} k_t = 0
$$
and 
\begin{align*}
&\text{if $\sigma^2=0$}, \quad \lim_{t\downarrow 0} \frac{t}{k_t} = 0,\\
&\text{if $\sigma^2>0$}, \quad \lim_{t\downarrow 0}
\frac{\sqrt{t}}{k_t} = 0. 
\end{align*}
Then, 
as $t\downarrow 0$,
it holds:
\begin{align}
\mathbb E[(e^{X_t} - e^{k_t})^+] = \mathbb E[(e^{\sigma W_t- \frac{\sigma^2 t}{2}}
 -e^{k_t})^+] + t \int_{(0,\infty)} (e^x-1)\nu(dx) +o(t)\label{finvar1.eq}
\end{align}
and 
\begin{align}
\mathbb E[(e^{-k_t}-e^{X_t})^+] = \mathbb E[(e^{-k_t} - e^{\sigma W_t- \frac{\sigma^2 t}{2}})^+] +  t \int_{(-\infty,0)} (1-e^x)\nu(dx) +o(t).\label{finvar2.eq}
\end{align}
\end{proposition}

\begin{remarks}
\begin{enumerate}[(i)]
\item Proposition~\ref{finvar}
implies that, in the absence of a Brownian component,
the call and put prices of options struck at 
$e^{k_t}$
and
$e^{-k_t}$,
respectively,
tend to zero at the rate equal to 
$t$
if $X$
has paths of 
finite variation
(cf. Remark~\eqref{rem:slower_than_t} after Theorem~\ref{infvar}).

\item The Blumenthal-Getoor indices of the positive and negative 
jump processes
$X^+$
and
$X^-$
of  
$X$,
defined in Section~\ref{subsubsec:BG_index},
are both smaller or equal to one
by the assumption in Proposition~\ref{finvar}.
Furthermore, unlike in the case of jumps of infinite variation, 
Proposition~\ref{finvar}
implies that the asymptotic behaviour of short-dated option 
prices (as maturity $t$ tends to zero)
does not depend up to order
$o(t)$
on the indices
$\textrm{BG}(X^+)$
and
$\textrm{BG}(X^-)$.
Hence, the same will hold for the short-dated implied
volatility (cf. Corollary~\ref{volimp.cor}).

\item It should be stressed that the proof of Proposition~\ref{finvar},
given in Section~\ref{sec:Proof_Prop_finvar},
is fundamentally different from that of 
Theorem~\ref{infvar},
as it relies on the path-wise version of the It\^o-Tanaka formula for the processes 
of finite variation, which cannot be applied in the context of
Theorem~\ref{infvar}.
\end{enumerate}
\end{remarks}

\section{Asymptotic behaviour of implied volatility}
\label{sec:Asym_IVOL}

The value 
$C^{\textrm{BS}}(t,k,\sigma)$
of the European call option with strike
$e^k$
(for any
$k\in\RR$)
and expiry 
$t$
under a Black-Scholes 
model
(with 
log-spot 
$X_t=\sigma W_t-t\sigma^2/2$
of constant volatility 
$\sigma>0$)
is given by 
the Black-Scholes formula
\begin{eqnarray}
\label{eq:BS_Formula}
C^{\textrm{BS}}(t,k,\sigma) & = & N(d_+)-e^kN(d_-),\qquad
\text{where}
\quad
d_{\pm}=-\frac{k}{\sigma\sqrt{t}}\pm
\frac{\sigma\sqrt{t}}{2}
\end{eqnarray}
and
$N(\cdot)$
is the standard normal cumulative distribution function.
The price of a put option with the same strike and maturity
is given by
$P^{\textrm{BS}}\left(t,k,\sigma\right)= e^k N(-d_+)-N(-d_-)$.
Let 
$S$
be a positive martingale, with
$S_0=1$,
that models a risky security
and denote by 
\begin{equation}
\label{eq:Call_Put_Gen_Mod}
C(t,k):= \EE\left[\left(S_t-e^k\right)^+\right]
\qquad \text{and}\qquad
P(t,k):= \EE\left[\left(e^k-S_t\right)^+\right]
\end{equation}
the prices of call and put options on
$S$
struck at
$e^k$
with maturity 
$t$,
respectively.
The \textit{implied volatility}
in the model 
$S$
for any log-strike 
$k\in\RR$
and
maturity 
$t>0$
is the unique positive number
$\wh \sigma(t,k)$
that satisfies the following equation in 
$\sigma$:
\begin{eqnarray}
\label{eq:DefOfImpiedVol}
C^{\textrm{BS}}\left(t,k,\sigma\right)=C(t,k).
\end{eqnarray}
Implied volatility is well-defined since 
the function
$\sigma\mapsto C^{\textrm{BS}}\left(t,k,\sigma\right)$
is strictly increasing
on the positive half-line
and the right-hand side of~\eqref{eq:DefOfImpiedVol} 
lies in the image of the Black-Scholes formula
by a simple no-arbitrage argument. 
Put-call parity, which holds since 
$S$
is a true martingale, implies the identity
$P^{\textrm{BS}}\left(t,k,\wh \sigma(t,k)\right)=P(t,k)$.

In order to study the limiting behaviour of the implied volatility
close to the at-the-money strike 
$1=e^0$ for short maturities, 
we define the following parameterisation of the log-strike
$k_t$:
\begin{eqnarray}
\label{eq:def_kt}
k_t &:=& \theta \left(t \log\frac{1}{t}\right)^{1/2},\qquad\text{where}\quad
\theta\in\RR\!\setminus\!\{0\}.
\end{eqnarray}
We can now define the implied volatility 
$\sigma_t:\RR\!\setminus\!\{0\}\to(0,\infty)$
as a function of 
$\theta$
in the asymptotic maturity-strike regime
$(t,k_t)$,
given 
by~\eqref{eq:def_kt},
for a short maturity
$t$:
\begin{eqnarray}
\label{eq:Simga_t_Def}
\sigma_t(\theta) & :=  & \wh \sigma\left(t,k_t\right).
\end{eqnarray}
The implied volatility 
$\sigma_t(\theta)$
is of interest in the context of processes
with jumps, because its limit 
$\sigma_0(\theta)$,
as
$t\downarrow0$,
exists and is finite for each 
$\theta$,
depends on both the jump and the diffusion components
of the process and can be computed explicitly in terms
of the parameters. In order to find the asymptotic 
behaviour of 
$\sigma_t(\theta)$,
we first state 
Theorem~\ref{ivexpand},
which relates the asymptotics of 
$\sigma_t(\theta)$
to the asymptotic behaviour of the
out-of-the-money option price 
\begin{equation}
\label{eq:Out_of_money_option}
I_t(\theta):=C(t,k_t)1_{\{\theta>0\}}+P(t,k_t)1_{\{\theta<0\}}
\end{equation}
under 
the model 
$S$
as
maturity 
$t$
tends to zero.

\begin{theorem}\label{ivexpand}
Let 
$S$ be a martingale model for a risky security with
$S_0=1$
and
$k_t$
a log-strike given in~\eqref{eq:def_kt}
for a fixed
$\theta\in\RR\!\setminus\!\{0\}$.
Let
$\wh C_t$ 
and
$\wh P_t$ 
be deterministic functions 
such that 
$C(t,k_t)\sim \wh C_t$ 
and
$P(t,k_t)\sim \wh P_t$ 
as $t\downarrow 0$, 
where
$C(t,k_t)$ 
and
$P(t,k_t)$ 
are given in~\eqref{eq:Call_Put_Gen_Mod},
and define
$\wh I_t(\theta):=\wh C_t1_{\{\theta>0\}}+\wh P_t1_{\{\theta<0\}}$.
Assume further that the 
out-of-the-money option price
$I_t(\theta)$,
given in~\eqref{eq:Out_of_money_option},
satisfies:
\begin{equation}
\label{eq:Assumption_IVOL}
\frac{1}{2}<
\liminf_{t\downarrow 0} \frac{\log I_t(\theta)}{\log t}\leq
\limsup_{t\downarrow 0} \frac{\log I_t(\theta)}{\log t}<\infty.
\end{equation}
Then the implied volatility 
$\sigma_t(\theta)$,
defined in~\eqref{eq:Simga_t_Def},
can be expressed by
\begin{align}
\sigma_t(\theta) &= \frac{|\theta|}{\sqrt{2L_t(\theta)-1}} + \frac{|\theta|
  \log\frac{\left(2L_t(\theta)-1\right)^{\frac{3}{2}}\sqrt{2\pi}}{|\theta|}}{\left(2L_t(\theta)-1\right)^{\frac{3}{2}}}\frac{1}{\log
\frac{1}{t}}
+ O\left(\frac{1}{\log^2 \frac{1}{t}}\right),\qquad \text{as}\quad
t\downarrow 0, \label{bigo}
\end{align}
and
\begin{align}
\sigma_t(\theta)&= \frac{|\theta|}{\sqrt{2\wh L_t(\theta)-1}} + \frac{|\theta|
  \log\frac{\left(2\wh L_t(\theta)-1\right)^{\frac{3}{2}}\sqrt{2\pi}}{|\theta|}}{\left(2\wh L_t(\theta)-1\right)^{\frac{3}{2}}}\frac{1}{\log
\frac{1}{t}}
+ o\left(\frac{1}{\log \frac{1}{t}}\right),\qquad \text{as}\quad
t\downarrow 0,
\label{littleo}
\end{align}
where $L_t(\theta):=J_t(I_t(\theta))$ and
$\wh L_t(\theta):=J_t(\wh I_t(\theta))$
are defined by the formula
$$
J_t(x):= \frac{\log x }{\log t} -\frac{\log\log\frac{1}{t}}{\log \frac{1}{t}} \quad \text{for any} \quad 
x,t>0.
$$ 
\end{theorem}

Before proceeding with the application and proof of Theorem~\ref{ivexpand},
we make the following remarks in order to place it in context.
\begin{remarks}
\begin{enumerate}[(i)]
\item In the Black-Scholes model with volatility $\sigma>0$,
the following well-known expansion of the call option price 
in the 
$(t,k_t)$
maturity-strike regime~\eqref{eq:def_kt} holds (e.g. a straightforward calculation
using~\cite[Eq.~(3.10)]{gao.lee.11} yields the expansion):
\begin{eqnarray}
\label{eq:BS_expansion}
C^{\textrm{BS}}(t,k_t,\sigma) = \frac{\sigma}{\sqrt{2\pi}} t^{\frac{1}{2} + \frac{\theta^2}{2\sigma^2}}
\left\{\frac{\sigma^2}{\theta^2}\frac{1}{\log\frac{1}{t}}- 3
  \frac{\sigma^4}{\theta^4}\frac{1}{\log^2\frac{1}{t}} + O\left(\frac{1}{\log^3\frac{1}{t}}\right)\right\}
  \qquad\text{as}\quad t\downarrow0.  
\end{eqnarray}
In particular 
we have
$\log C^{\textrm{BS}}(t,k_t,\sigma) =  (\frac{1}{2} + \frac{\theta^2}{2\sigma^2})\log t + o(\log(1/t))$
as
$t\downarrow0$
and hence the assumption in~\eqref{eq:Assumption_IVOL} is satisfied
in the Black-Scholes model.
\label{rem:BS_expansion}

\item Note that the log-strike
$k_t$
in~\eqref{eq:def_kt} satisfies the assumptions of 
Theorem~\ref{infvar}.
For any L\'evy process $X$
as in Theorem~\ref{infvar},
formula~\eqref{infvar.eq}
and Remark~\eqref{rem:BS_expansion} above
imply 
\begin{equation}
\log C(t,k_t) =  \min\left\{\frac{3-\alpha_+}{2}, \frac{1}{2} + \frac{\theta^2}{2\sigma^2}\right\} \log t + o(\log(1/t))
  \qquad\text{as}\quad t\downarrow0.  
\label{eq:Infvar_option_asym}
\end{equation}
Since the minimum of the constants in front of 
$\log t$
is clearly larger than 
$1/2$,
assumption~\eqref{eq:Assumption_IVOL}
of Theorem~\ref{ivexpand} is satisfied.
As we shall soon see, it is the balance (as a function of
$\theta$)
between the two constants in~\eqref{eq:Infvar_option_asym} that determines the value of the 
limiting smile
$\sigma_0(\theta)$.
\label{rem:InfVar}

\item 
Let a L\'evy process $X$ be as in 
Proposition~\ref{finvar}
(i.e. with jumps of finite variation).
Formulae~\eqref{finvar1.eq}
and~\eqref{eq:BS_expansion}
imply
that the call option price
$C(t,k_t)$
under the model 
$S$
has the following asymptotic behaviour
\begin{equation}
\log C(t,k_t) =  \min\left\{1, \frac{1}{2} + \frac{\theta^2}{2\sigma^2}\right\} \log t + o(\log(1/t))
  \qquad\text{as}\quad t\downarrow0.  
\label{eq:Finvar_option_asym}
\end{equation}
In particular note that 
assumption~\eqref{eq:Assumption_IVOL}
is satisfied and that, in the case of jumps with 
finite variation, the constant in front of 
$\log t$
does not depend on the L\'evy measure but solely 
on the diffusion component of the model. 
\label{rem:FinVar}

\item In~\cite{gao.lee.11} the authors present a general result,
which translates  the asymptotic behaviour of the option prices, 
in a generic maturity-strike regime, to the asymptotics of the corresponding
implied volatilities. Unfortunately the results in~\cite{gao.lee.11}
do not apply in the regime 
$(t,k_t)$,
for 
$k_t$
in~\eqref{eq:def_kt}, 
since the standing assumption
of~\cite{gao.lee.11},
$\max\{0,\log(1/k_t)\}=o(\log(1/C(t,k_t)))$
(see~\cite[Eq.~(4.3)]{gao.lee.11}),
is not satisfied in our setting by~\eqref{eq:Infvar_option_asym}
and~\eqref{eq:Finvar_option_asym}.
We therefore have to establish Theorem~\ref{ivexpand},
which is applicable in our context 
as remarked in~\eqref{rem:InfVar} and~\eqref{rem:FinVar}
above.
\label{rem:Not_in_Gao_Lee}
\end{enumerate}
\end{remarks}

Before proving Theorem~\ref{ivexpand},
we apply it, together with Theorem~\ref{infvar}
and Proposition~\ref{finvar},
to derive the main asymptotic formula 
of the paper. 

\begin{corollary}\label{volimp.cor}
Let
$X$
be a L\'evy process
with the jump measure 
$\nu$
and
the Gaussian component
$\sigma^2\geq0$.
Pick
$\theta \in\RR\setminus\{0\}$, 
let
$k_t$ 
be the log-strike  
from~\eqref{eq:def_kt}
and let
$\sigma_t(\theta)$
be the implied volatility 
defined in~\eqref{eq:Simga_t_Def}.
Then the following statements hold.
\begin{enumerate}[(a)]
\item Let $X$ be a L\'evy process satisfying the assumptions of Theorem~\ref{infvar}.
\label{item:(a)_Cor}
Then the implied volatility 
$\sigma_t(\theta)$
takes the form
\begin{align}
\label{assform}
\sigma_t(\theta) = \left\{ \begin{array}{ll}
\frac{\pm\theta}{\sqrt{2-\alpha_\pm}}\left[1+I_\pm(t,\theta)\right]+ o\left(\frac{1}{\log \frac{1}{t}}\right), 
&\textrm{if $\pm\theta\geq\sigma \sqrt{2-\alpha_\pm}$ and $c_\pm>0$},\\
\sigma+ o\left(\frac{1}{\log \frac{1}{t}}\right), &\textrm{if $0<\pm\theta<\sigma \sqrt{2-\alpha_\pm}$ and $c_\pm>0$},
\end{array} \right.
\qquad\text{as $t\downarrow0$,}
\end{align}
where
\begin{equation}
\label{eq:f_pm}
I_\pm(t,\theta):= \frac{3-\alpha_\pm}{2(2-\alpha_\pm)}\frac{\log \log \frac{1}{t}}{\log \frac{1}{t}} 
+\frac{1}{(2-\alpha_\pm)}
\log\left(\frac{(2-\alpha_\pm)^{\frac{3}{2}} c_\pm\sqrt{2\pi }}{|\theta|^{{\alpha_\pm}} (\alpha_\pm-1)}\right)\frac{1}{\log\frac{1}{t}}, 
\qquad\text{for $t>0$,}
\end{equation}
and the sign
$\pm$
denotes either
$+$
or
$-$
throughout the formulae in~\eqref{assform} and~\eqref{eq:f_pm}.
In particular, the limiting smile 
$\sigma_0(\theta):=\lim_{t\downarrow 0} \sigma_t(\theta)$
exists 
for any
$\theta \in\RR\setminus\{0\}$
and
takes the form
$$
\sigma_0(\theta)=
\max\left\{\frac{\pm\theta}{\sqrt{2-\alpha_\pm}},
\sigma
\right\}\qquad\text{if $c_\pm>0$}.
$$

\item 
\label{item:(b)_Cor}
Let a L\'evy process $X$ be as in 
Proposition~\ref{finvar}
and let 
$\gamma_+,\gamma_-\geq0$
be equal to the following integrals
\begin{eqnarray*}
 \gamma_+:=\int_{(0,\infty)} (e^x-1)\nu(dx),\qquad
 \gamma_-:=\int_{(-\infty,0)} (1-e^x)\nu(dx).
\end{eqnarray*}
Then the implied volatility 
$\sigma_t(\theta)$
for short maturity 
$t$
is given by
\begin{align}
\label{eq:finv_var_assform}
\sigma_t(\theta) = \left\{ \begin{array}{ll}
\pm\theta\left[1+F_\pm(t,\theta)
\right]+ o\left(\frac{1}{\log \frac{1}{t}}\right), 
&\textrm{if $\pm\theta\geq\sigma$  and $\gamma_\pm>0$},\\ 
\sigma+ o\left(\frac{1}{\log \frac{1}{t}}\right), 
&\textrm{if $0<\pm\theta<\sigma$  and $\gamma_\pm>0$},
\end{array} \right.
\qquad\text{as $t\downarrow0$,}
\end{align}
where
\begin{equation}
\label{eq:F_pm}
F_\pm(t,\theta):= \frac{\log \log \frac{1}{t}}{\log \frac{1}{t}} 
+ \log\left(\frac{\gamma_\pm\sqrt{2\pi }}{|\theta|}\right)\frac{1}{\log\frac{1}{t}}, 
\qquad\text{for $t>0$,}
\end{equation}
and $\pm$
denotes either
$+$
or
$-$
throughout the formulae in~\eqref{eq:finv_var_assform} and~\eqref{eq:F_pm}.
The limit of the implied volatility smile as maturity tends to zero, 
$\sigma_0(\theta):=\lim_{t\downarrow 0} \sigma_t(\theta)$,
exists 
for 
$\theta \in\RR\setminus\{0\}$
and is equal to
$$
\sigma_0(\theta)=
\max\left\{\pm\theta, \sigma \right\}\qquad\text{if $\gamma_\pm>0$}.
$$
\end{enumerate}
\end{corollary}

\begin{remarks}
\begin{enumerate}[(i)]
\item Recall display~\eqref{halpha} in Theorem~\ref{infvar}
and note 
that the assumptions 
$c_+>0$
and 
$c_->0$
of Corollary~\ref{volimp.cor}~\eqref{item:(a)_Cor}
mean that,
as 
$x\downarrow0$,
the tails 
around zero
of the L\'evy  measure 
$\nu$
of
$X$
behave as
$\nu((x,\infty))\sim c_+x^{-\alpha_+}$
and
$\nu((-\infty,-x))\sim c_- x^{-\alpha_-}$.
Note further that, 
once we have identified the precise rate of the
tail behaviour of 
$\nu$
at zero,
the constants
$c_+$
and
$c_-$
do not feature in the limiting formula
$\sigma_0(\theta)$.

\item 
The assumption  
$\gamma_\pm>0$
in
Corollary~\ref{volimp.cor}~\eqref{item:(b)_Cor}
ensures that the process
$X$
has positive jumps
when
$\theta>0$
and negative jumps
when
$\theta<0$
as we are only interested in the asymptotic behaviour
of the implied volatility in the presence of jumps.
\end{enumerate}
\end{remarks}

In Section~\ref{sec:Proof_Cor}
we derive Corollary~\ref{volimp.cor}
from Theorem~\ref{ivexpand}
and in  Section~\ref{sec:Proof_Thm_ivexpand}
we establish 
Theorem~\ref{ivexpand}.

\subsection{Proof of Corollary~\ref{volimp.cor}} \textbf{(a)} Assume first that 
\label{sec:Proof_Cor}
$\sigma\sqrt{2-\alpha_+} > \theta>0$. 
Define
$ \wh C_t := C^{\textrm{BS}}(t,k_t,\sigma)$
and note that~\eqref{eq:BS_expansion},
the definition of 
$k_t$ in~\eqref{eq:def_kt}
and~\eqref{infvar.eq}
of Theorem~\ref{infvar} imply
\begin{equation}
\label{eq:sim_relation}
C(t,k_t)\sim \wh C_t,
\qquad\text{and hence}\quad
\frac{\log C(t,k_t)}{\log t}=
\frac{\log \wh C_t}{\log t}+
o\left(\frac{1}{\log \frac{1}{t}}\right),
\qquad\text{as $t\downarrow0$,}
\end{equation}
where
$C(t,k_t)$
denotes the call option price with maturity 
$t$
and strike 
$e^{k_t}$
under the exponential L\'evy model
$e^X$.
Assumption~\eqref{eq:Assumption_IVOL}
of Theorem~\ref{ivexpand}
is therefore satisfied
by Remark~\eqref{rem:BS_expansion}
after Theorem~\ref{ivexpand}. 
The formula for 
$\wh L_t(\theta) = 
\log \wh C_t/\log t- (\log \log \frac{1}{t})\log \frac{1}{t}$ 
takes the form
\begin{equation}
\label{eq:L_theta_BS}
\wh L_t(\theta) = 
\frac{1}{2} +
\frac{\theta^2}{2\sigma^2} - \log\left(\frac{\sigma^3}{\theta^2 \sqrt{2\pi}}\right)
\frac{1}{\log \frac{1}{t}}  +
o\left(\frac{1}{\log \frac{1}{t}}\right),
\qquad\text{as $t\downarrow0$,}
\end{equation}
The 
formula in~\eqref{littleo}
of
Theorem~\ref{ivexpand},
together with~\eqref{eq:L_theta_BS}
and the Taylor expansions 
in $\log(1/t)$
as
$t\downarrow0$
\begin{align*}
\frac{\theta}{\sqrt{2\wh L_t(\theta)-1}} & = 
\sigma
\left[
1+\frac{\sigma^2}{\theta^2}   \log\left(\frac{ \sigma^3}{\theta^2\sqrt{2\pi}}\right) \frac{1}{\log \frac{1}{t}}\right]
+ o\left(\frac{1}{\log \frac{1}{t}}\right), 
\end{align*}
\begin{align*}
\frac{\theta
  \log\frac{\left(2\wh L_t(\theta)-1\right)^{\frac{3}{2}}\sqrt{2\pi}}{\theta}}{\left(2\wh L_t(\theta)-1\right)^{\frac{3}{2}}}\frac{1}{\log
\frac{1}{t}} 
& = 
\frac{\sigma^3 \log\frac{\theta^2\sqrt{2\pi}}{\sigma^3}}{\theta^2}\frac{1}{\log \frac{1}{t}} 
+ o\left(\frac{1}{\log \frac{1}{t}}\right),
\end{align*}
yield the formula in~\eqref{assform}.

In the case
$\sigma\sqrt{2-\alpha_+} \leq \theta$,
the relation~\eqref{eq:sim_relation}
is satisfied by
$\wh C_t := \frac{tk_t^{1-\alpha_+}c_+}{\alpha_+-1}$.
This follows directly from the definition of
$k_t$ in~\eqref{eq:def_kt}
and 
Theorem~\ref{infvar} 
(see formula~\eqref{infvar.eq}).
An analogous argument as the one above shows that
in this case the assumptions
of Theorem~\ref{ivexpand}
are also satisfied.
By definition
of
$\wh L_t(\theta)$
in Theorem~\ref{ivexpand},
we find
$$
2\wh L_t(\theta)-1 = (2-\alpha_+)\left[
1-\frac{3-\alpha_+}{2-\alpha_+}  \frac{\log \log \frac{1}{t}}{\log \frac{1}{t}}
- \frac{2}{2-\alpha_+} \log\left(\frac{\theta^{1-\alpha_+} c_+}{\alpha_+-1}\right) \frac{1}{\log \frac{1}{t}}  
\right].
$$
By Taylor's formula the following asymptotic relations 
hold
as
$t\downarrow0$:
\begin{align*}
\frac{\theta}{\sqrt{2\wh L_t(\theta)-1}} & = 
\frac{\theta}{\sqrt{2-\alpha_+}}\left[
1+\frac{3-\alpha_+}{2(2-\alpha_+)}  \frac{\log \log \frac{1}{t}}{\log \frac{1}{t}}
+ \frac{1}{2-\alpha_+} \log\left(\frac{\theta^{1-\alpha_+} c_+}{\alpha_+-1}\right) \frac{1}{\log \frac{1}{t}}\right]
+ o\left(\frac{1}{\log \frac{1}{t}}\right), 
\end{align*}
and
\begin{align*}
\frac{\theta
  \log\frac{\left(2\wh L_t(\theta)-1\right)^{\frac{3}{2}}\sqrt{2\pi}}{\theta}}{\left(2\wh L_t(\theta)-1\right)^{\frac{3}{2}}}\frac{1}{\log
\frac{1}{t}} 
& = 
\frac{\theta \log\frac{\left(2-\alpha_+\right)^{\frac{3}{2}}\sqrt{2\pi}}{\theta}}{\left(2-\alpha_+\right)^{\frac{3}{2}}}\frac{1}{\log \frac{1}{t}} 
+ o\left(\frac{1}{\log \frac{1}{t}}\right).
\end{align*}
Substituting these expressions into~\eqref{littleo}
establishes the formula in~\eqref{assform}.

Assume now that 
$-\sigma\sqrt{2-\alpha_-} < \theta<0$. 
Define
$\wh P_t:= P^{\textrm{BS}}(t,k_t,\sigma)$,
where
$P^{\textrm{BS}}(t,k_t,\sigma)$
is the put option price in the Black-Scholes model,
and recall the well-known put-call symmetry
\begin{equation}
\label{eq:Put_Call_Sy_BS}
P^{\textrm{BS}}(t,k_t,\sigma)=
e^{k_t}C^{\textrm{BS}}(t,-k_t,\sigma),
\end{equation}
which holds since 
the laws of minus the log-spot
under the share measure (i.e. the pricing measure where 
the risky asset is a numeraire)
and the log-spot under the risk-neutral measure 
(i.e. the measure where the riskless asset is the numeraire)
coincide.
Analogous to the case above,~\eqref{eq:BS_expansion}
with the put-call symmetry,
the definition of 
$k_t$ in~\eqref{eq:def_kt}
and~\eqref{infvar.eq_Put}
of Theorem~\ref{infvar} imply
\begin{equation}
\label{eq:sim_relation_Put}
P(t,k_t)\sim \wh P_t,
\qquad\text{and hence}\quad
\frac{\log P(t,k_t)}{\log t}=
\frac{\log \wh P_t}{\log t}+
o\left(\frac{1}{\log \frac{1}{t}}\right),
\qquad\text{as $t\downarrow0$,}
\end{equation}
where
$P(t,k_t)$
is 
the put option price 
under the exponential L\'evy model
$e^X$.
Therefore the assumptions of Theorem~\ref{ivexpand}
are satisfied and 
$\wh L_t(\theta)$
takes the form~\eqref{eq:L_theta_BS}.
Note that the right-hand side of~\eqref{eq:L_theta_BS}
depends solely on the even powers of 
$\theta$
and hence the 
fact
$\theta<0$
does not influence the asymptotic behaviour 
of
$\wh L_t(\theta)$.
The proof of formula~\eqref{assform} 
now follows in the same way as in the call case above.

In the case
$-\sigma\sqrt{2-\alpha_-} \geq \theta$
we define
$\wh P_t := \frac{t(-k_t)^{1-\alpha_+}c_+}{\alpha_+-1}$.
Under this assumption,
the relation~\eqref{eq:sim_relation_Put}
is satisfied by~\eqref{infvar.eq_Put}
of Theorem~\ref{infvar} and the rest of the proof
follows along the same lines as in the case
$\sigma\sqrt{2-\alpha_+} \leq \theta$.
This proves formula~\eqref{assform}.

\noindent \textbf{(b)} The proof of part~(b) of the corollary is based on 
Proposition~\ref{finvar} and Theorem~\ref{ivexpand}.
The steps are analogous to the ones in the proof of part~(a):
\begin{eqnarray*}
& \text{if}\quad\sigma > \theta>0,\quad\text{define}\>\> \wh C_t := C^{\textrm{BS}}(t,k_t,\sigma);& 
\text{if}\quad\sigma \leq \theta,\quad\text{define}\>\> \wh C_t := t \gamma_+; \\
&\text{if}\quad-\sigma < \theta<0,\quad\text{define}\>\> \wh P_t := P^{\textrm{BS}}(t,k_t,\sigma);&  
\text{if}\quad-\sigma \geq \theta,\quad\text{define}\>\> \wh P_t := t \gamma_-. 
\end{eqnarray*}
The details of the calculations are left to the reader.

\subsection{Proof of Theorem~\ref{ivexpand}}
\label{sec:Proof_Thm_ivexpand}
We first assume that 
$\theta>0$.
Equality~\eqref{eq:BS_expansion}
implies the following
\begin{align*}
\frac{\log C^{\textrm{BS}}(t,k_t,s)}{\log t} -\frac{\log\log\frac{1}{t}}{\log \frac{1}{t}} 
& = \frac{1}{2} + \frac{\theta^2}{2s^2}   - \frac{1}{\log \frac{1}{t}}\log \frac{s^3}{\theta^2\sqrt{2\pi}}
+  3
  \frac{s^2}{\theta^2}\frac{1}{\log^2\frac{1}{t}} +
  O\left(\frac{1}{\log^3\frac{1}{t}}\right)
\end{align*}
as
$t\downarrow0$
for any
$s>0$.
Define
$$
F(s,z) :=  -z{\log C^{\textrm{BS}}(e^{-1/z},k_{e^{-1/z}},s)} +z\log z
$$
and note that 
$F(s,z)$
corresponds to the left-hand side of the above formula with the
change of variable $z = \frac{1}{\log \frac{1}{t}}$. The 
expansion shows that $F(s,z)$ is regular as $z\to 0$ and
the following equality holds
$$
F(s,z) = \frac{1}{2} + \frac{\theta^2}{2s^2}   - z\log \frac{s^3}{\theta^2\sqrt{2\pi}}
+  3
  \frac{s^2}{\theta^2}z^2 +
  O\left(z^3\right).
$$
The expansion for the inverse mapping can be deduced from this
expression as follows. To keep the formulae simple, we give the
expansion up to $O(z^2)$: 
$$
F(s,z) = a(s) + zb(s) + O(z^2), \qquad\text{where}\quad
a(s)=\frac{1}{2} + \frac{\theta^2}{2s^2},\quad
b(s)=\log \frac{s^3}{\theta^2\sqrt{2\pi}}.
$$
Denote by  
$F^{-1}(y,z)$ 
the unique positive solution of the equation 
$F(s,z)=y$,
where
$y$
equals
$J_{e^{-1/z}}(x)$
(see the statement of Theorem~\ref{ivexpand}
for the definition of 
$J_t(x)$)
and 
$x$
is any arbitrage-free
call
option price with maturity 
$e^{-1/z}$
and
strike
$k_{e^{-1/z}}$.
The uniqueness of the quantity 
$F^{-1}(y,z)$ 
is equivalent to the fact that the
implied volatility is a well defined quantity.

An approximate expression for 
$y$
is given by
$$
y = a(F^{-1}(y,z)) + z b(F^{-1}(y,z)) + O(z^2)
$$
and hence we find
$$
a^{-1}(y) = a^{-1}(a(F^{-1}(y,z)) + z b(F^{-1}(y,z)) + O(z^2)).
$$
Using the regularity of the coefficient $a$
in the neighbourhood of the point
$F^{-1}(y,z)>0$,
we can expand the inverse $a^{-1}$ 
around the point $a(F^{-1}(y,z))$
as follows:
$$
a^{-1}(y) = F^{-1}(y,z) + (a^{-1})' (a(F^{-1}(y,z))) b(F^{-1}(y,z))z + O(z^2).
$$
In view of this expression, and using once again the regularity of the
coefficients $a$ and $b$, we can replace $F^{-1}(y,z)$ with
$a^{-1}(y)$ in the second term, obtaining
$$
a^{-1}(y) = F^{-1}(y,z) + (a^{-1})' (y) b(a^{-1}(y))z + O(z^2).
$$ 
Hence,
the following asymptotic equalities  hold true:
\begin{align*}
F^{-1}(y,z) &= a^{-1}(y) - \frac{b(a^{-1}(y))}{a'(a^{-1}(y))}z +
O(z^2)\\
& = \frac{\theta}{\sqrt{2y-1}} + \frac{\theta
  \log\frac{(2y-1)^{\frac{3}{2}}\sqrt{2\pi}}{\theta}}{(2y-1)^{\frac{3}{2}}}z + O(z^2).
\end{align*}
Substituting the expression for $F$, we find an expansion for the
implied volatility 
$\sigma_t(\theta)$
given in~\eqref{bigo}.
Now, $ C(t,k_t)\sim \wh C_t$ implies that $\frac{\log C(t,k_t)}{\log t}
- \frac{\log \wh C_t }{\log t} = o(\log^{-1} t^{-1})$. 
Since all the coefficients in expansion~\eqref{bigo}
are regular, the additional term arising from this difference may be ignored 
in an expansion up to order
$o(\log^{-1} t^{-1})$ 
and~\eqref{littleo} follows. 

The formulae in the theorem in the case 
$\theta<0$ 
will be established by applying the result for 
the positive log-strike under the share measure.
More precisely, let
$\PP$
denote the original risk-neutral measure under which
the process
$S$
is a positive martingale started at one.
For each time 
$t$,
we define the share measure 
$\wt \PP$
on the 
$\sigma$-algebra 
$\mathcal F_t$
of events that can occur up to time
$t$
via its Radon-Nikodym derivative
$\frac{d\wt \PP}{d\PP}|_{\mathcal F_t}:=S_t$
and note that the following relationship holds
for any log-strike
$k\in\RR$:
\begin{equation}
\label{eq:put_call_sym}
P(t,k)=\EE\left[(e^k-S_t)^+\right]
=e^{k}\wt \EE\left[(S_t^{-1}-e^{-k})^+\right]=
e^k \wt C(t,-k),
\end{equation}
where 
$\wt C(t,-k):=\wt \EE\left[(S_t^{-1}-e^{-k})^+\right]$
denotes the expectation under the share measure 
$\wt \PP$
of a call payoff with strike
$e^{-k}$,
where the evolution of the risky asset is given by 
$S^{-1}$.
Note that 
$S^{-1}$
is a positive martingale started at one
under
$\wt \PP$
and hence 
$\wt C(t,-k)$
represents and arbitrage-free call option price. 
Furthermore, the put-call symmetry formula in the
Black-Scholes model (see~\eqref{eq:Put_Call_Sy_BS})
and the equality in~\eqref{eq:put_call_sym}
mean that the implied volatility 
$\wh\sigma(t,k)$
defined by the put price
$P(t,k)$
coincides with the implied volatility
$\wh{\wt\sigma}(t,-k)$
defined by the call price
$\wt C(t,-k)$
(see beginning of Section~\ref{sec:Asym_IVOL}
for the definition of 
$\wh \sigma(t,k)$).

Note that, since 
$\theta<0$,
we now have
$-k_t>0$
and 
$\sigma_t(\theta)=\wt \sigma_t(-\theta)$,
where
$\wt \sigma_t(-\theta)$
denotes 
$\wh{\wt\sigma}(t,-k_t)$.
In order to apply the formula in~\eqref{bigo} 
to 
$\wt C(t,-k_t)$,
we have to ensure that assumption~\eqref{eq:Assumption_IVOL}
is satisfied.
Since~\eqref{eq:Assumption_IVOL} 
holds for 
$P(t,k_t)$
and
$k_t=o(\log t)$,
the equality  in~\eqref{eq:put_call_sym}
implies~\eqref{eq:Assumption_IVOL} 
for 
$\wt C(t,-k_t)$.
Therefore formula~\eqref{bigo}
gives an asymptotic expansion 
of
$\sigma_t(\theta)=\wt \sigma_t(-\theta)$
in terms of 
$\wt L_t(-\theta):=J_t(\wt C(t,-k_t))$.
Since equality~\eqref{eq:put_call_sym}
implies
$$
L_t(\theta)=\wt L_t(-\theta)-\theta\sqrt{t/\log(1/t)} 
=\wt L_t(-\theta)+O\left(\frac{1}{\log^2\frac{1}{t}}\right)\qquad\text{as $t\downarrow0$}
$$
and the two leading order terms in~\eqref{bigo}
are regular in 
$\wt L_t(-\theta)$,
the asymptotic expansion in~\eqref{bigo}
also holds when 
$\wt L_t(-\theta)$
is replaced by
$L_t(\theta)$.
The formula in~\eqref{littleo}
now follows by the same argument as in the
case of the positive log-strike. 
This concludes the proof of the theorem.

\section{Numerical results}
\label{sec:Numerical}

In this section, we present some numerical illustrations for the
convergence results discussed in Section~\ref{sec:Asym_IVOL}.
We focus on the generalised tempered stable L\'evy process 
$X$
with L\'evy density
\begin{eqnarray}
\label{eq:tempered_Stable}
\nu(x) = \frac{c_+ e^{-\lambda_+ x}}{|x|^{1+\alpha_+}}1_{\{x>0\}} +
\frac{c_- e^{-\lambda_- |x|}}{|x|^{1+\alpha_-}}1_{\{x<0\}}. 
\end{eqnarray}
This class of processes includes the widely used CGMY models 
(see e.g.~\cite{finestructure}).
For this process, the price of a European call option with pay-off
$(S_0e^{X_t}-K)^+$ at time $t$ can be computed as
\begin{align}
\frac{K}{2\pi} \int_{\mathbb R} \left(\frac{K}{S_0}\right)^{iu-R}\frac{\phi_t(-u-iR) }{(R-iu) (R-1-iu)}
du,\label{fourierinv}
\end{align}
where $\phi_t$ is the characteristic function of $X_t$ and $R>1$
(see e.g.~\cite{carrmadan} or~\cite{tankov.09}). 
We compute the integral in~\eqref{fourierinv} with an
adaptive integration algorithm. 

\subsection{Testing the algorithm} To 
ensure that the prices returned by our algorithm are correct, we
first compare them to the values computed in \cite{wang2007robust}
with their approximate ``fixed point'' algorithm (PDE
discretisation). The following table shows that the values we obtain
are very similar with the small discrepancy probably due to the
discretisation error of~\cite{wang2007robust}.\\ 

\begin{center}
\begin{tabular}{|c|c|c|c|c|c|c|c|c|c|}
\hline
$S$ & $K$ & $T$ & $r$ & $c=c_+=c_-$ & $\lambda_+$ & $\lambda_-$ & $\alpha=\alpha_+=\alpha_-$ &
Value (\cite{wang2007robust}) & Our value \\
\hline
$90$ & $98$ & $0.25$ & $0.06$ & $16.97$ & $29.97$ & $7.08$ & $0.6442$ & $16.212578$ & $16.211904$ \\
$90$ & $98$ & $0.25$ & $0.06$ & $0.42$ & $191.2$ & $4.37$ & $1.0102$ & $2.2307031$ & $2.2306558$ \\
$10$ & $10$ & $0.25$ & $0.1$ & $1$ & $9.2$ & $8.8$ & $1.8$ & $4.3714972$ & $4.3898433$\\
\hline
\end{tabular}
\end{center}

\subsection{Convergence of the at-the-money (ATM) options}
In this section we fix the parameters of the tempered stable process at
\begin{equation}
\label{eq:Parm_Values}
c_+ = c_- = c= 1,\quad \lambda_+ = \lambda_- = 3,\quad \alpha_+ = \alpha_- = \alpha = 1.5
\end{equation}
and $S_0=1$. First we analyse the rate of convergence to zero
of the ATM options. It follows from the results in~\cite{muhle2011small} 
that the ATM option price satisfies
$$
\mathbb E [(e^{X_t} - 1)^+] \sim t^{1/\alpha} \mathbb E[(Z^*)^+],
$$
where $Z^*$ is a stable random variable with the L\'evy density 
$\frac{c}{|x|^{1+\alpha}}$.  
Furthermore it is known that
$$
\mathbb E[(Z^*)^+] = \frac{(2c)^{1/\alpha}}{\pi} \Gamma(1-1/\alpha)\left(-\Gamma(\alpha)\cos\frac{\pi\alpha}{2}\right)^{1/\alpha}=:C.
$$
Figure~\ref{mkn} plots the dependence of the normalised option price
$t^{-1/\alpha}\mathbb E [(e^{X_t}-1)^+]$ 
and the normalised ``Bachelier'' price 
$t^{-1/\alpha}\mathbb E [(X_t)^+] $ on 
$\log t$,
i.e. on time to maturity expressed on the log-scale. 
The horizontal line in Figure~\ref{mkn} corresponds to the value of the constant $C$. 
The desired convergence is clearly visible. 


\begin{figure}
\centerline{\includegraphics[width = 0.6\textwidth]{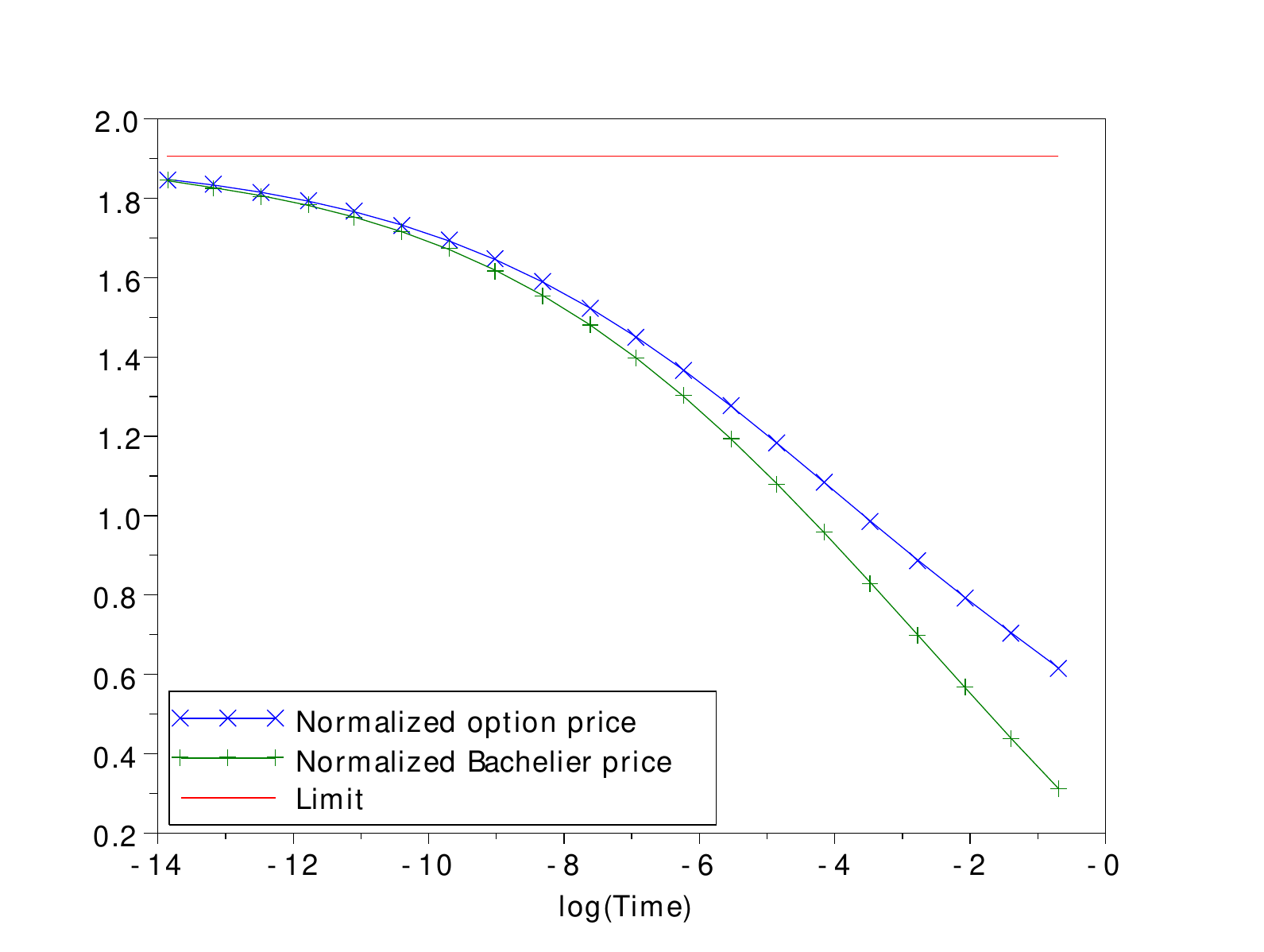}}
\caption{Convergence of the re-normalised price for ATM options:
the parameter values of the tempered stable process, which
models the log-stock are given in~\eqref{eq:Parm_Values}.}
\label{mkn}
\end{figure}

\subsection{Convergence of option prices with variable strike}
In this section we investigate numerically the convergence of the 
out-of-the-money (OTM) option prices given in Theorem~\ref{infvar}.
The parameter values for the underlying process are given in~\eqref{eq:Parm_Values}.
Note that in the case of the tempered stable L\'evy process 
with L\'evy density~\eqref{eq:tempered_Stable},
the limits in~\eqref{halpha} of Theorem~\ref{infvar}
take the form
$$
\lim_{x\downarrow0}x^\alpha\nu((x,\infty))=
\lim_{x\downarrow0}x^\alpha\nu((-\infty,-x))=
\frac{c}{\alpha}.
$$

Figure~\ref{cvgprice} shows the dependence of the normalised option 
and ``Bachelier'' prices,
respectively given by 
$$\frac{\mathbb E[(e^{X_t} - e^{k_t})^+]}{t k_t^{1-\alpha}}\qquad\text{ and }\qquad
\frac{\mathbb E[(X_t - k_t)^+]}{t k_t^{1-\alpha}},$$ 
on time to maturity in log-scale, where 
$$k_t = t^{1/\alpha'} 
\qquad\text{with}\qquad
\alpha' = 1.9.$$  
The horizontal dotted line shows the limiting value 
$\frac{c}{\alpha(\alpha-1)} = \frac{4}{3}$  
predicted by Theorem~\ref{infvar}.

Similarly, Figure~\ref{tlogt} plots the dependence on time to maturity 
(on the log-scale) of the normalised option price 
$$\frac{\mathbb E[(e^{X_t} - e^{k_t})^+]}{t k_t^{1-\alpha}}\qquad 
\text{for}\quad k_t = \theta \sqrt{t \log \frac{1}{t}}\qquad\text{and}\qquad 
\theta\in\{0.1,0.2,0.3\}.$$ 
As in Figure~\ref{cvgprice}, the limiting horizontal dotted line 
is given be $\frac{c}{\alpha(\alpha-1)} = \frac{4}{3}$.  

\begin{figure}
\centerline{\includegraphics[width = 0.6\textwidth]{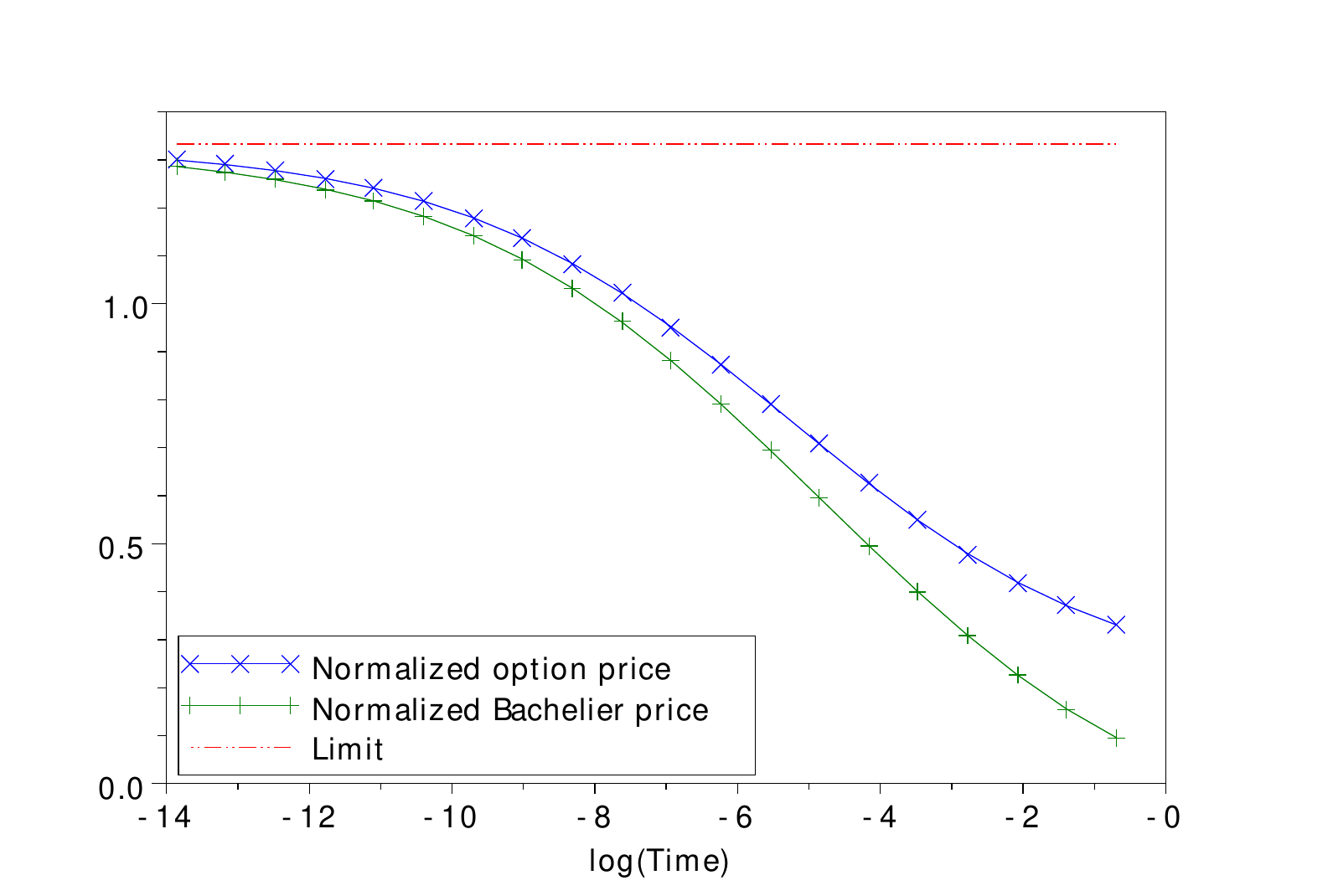}}
\caption{Convergence of the re-normalised price for OTM options: $k_t = t^{\frac{1}{\alpha'}}$
with $\alpha'=1.9$
and the other parameters of the process 
$X$ are given in~\eqref{eq:Parm_Values}. }
\label{cvgprice}
\end{figure}

\begin{figure}
\centerline{\includegraphics[width = 0.6\textwidth]{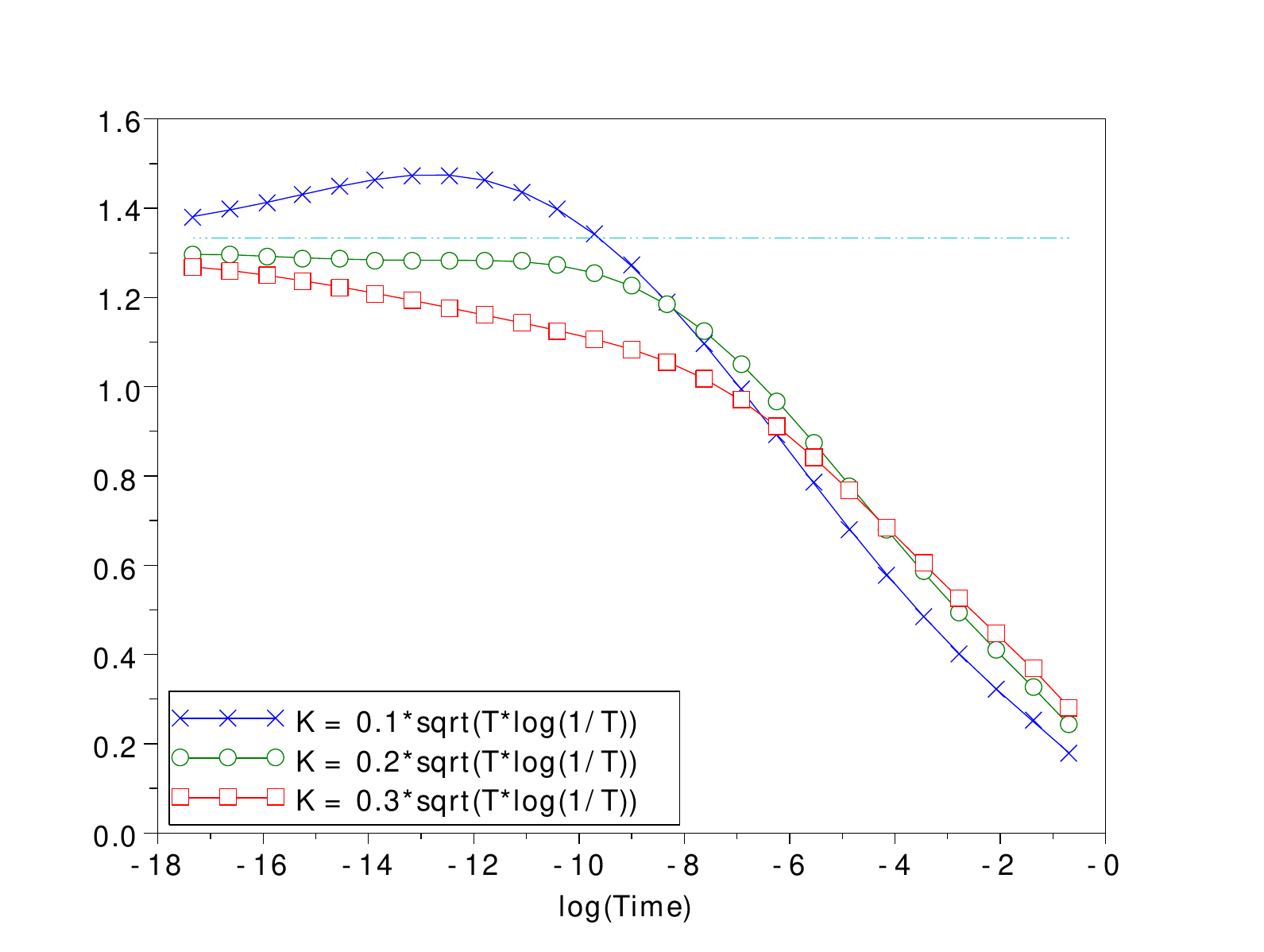}}
\caption{Convergence of the re-normalised price for OTM options: $k_t = \theta \sqrt{t \log \frac{1}{t}}$
with 
$\theta\in\{0.1,0.2,0.3\}$
and other parameters given in~\eqref{eq:Parm_Values}.}
\label{tlogt}
\end{figure}

\subsection{Convergence of the implied volatilities to the limiting smile}
\label{subsec:IVol_Asym}
In this section we illustrate the convergence of the implied
volatility (expressed as function of the re-normalised strike $\theta$)
to the limit 
$\sigma_0(\theta)$
given in Corollary~\ref{volimp.cor}.
In order to test the formula both with and without the diffusion
component we fix two models: the first is a pure jump
tempered stable L\'evy process 
with the following parameter values 
$$c_+ = c_- = c= 0.01,\qquad \lambda_+ = \lambda_- = 3,\qquad 
\alpha_+ = \alpha_- = \alpha = 1.5,$$ 
which correspond to the unit annualised volatility of about $14\%$.
The second model is the same tempered stable process
with an added diffusion component of 
volatility 
$\sigma = 0.2$. 

Recall that the limiting formula for positive 
$\theta$
is
$\sigma_0(\theta)=\max\{\sigma,\frac{\theta}{\sqrt{2-\alpha}}\}$.
Figure~\ref{ivol} plots the right wing of the implied volatility
smile (as function of $\theta$) for
different times to maturity when a diffusion component is present
(left graph) and diffusion component is absent (right graph),
together with the limiting shape 
$\sigma_0(\theta)$.
The convergence to the limit is
visible in both graphs but slow, because the error terms in 
Corollary~\ref{volimp.cor} are logarithmic in time. 
Nevertheless, the following observations can be made already at 
``not such small'' times:
\begin{itemize}
\item The smile is remarkably stable in time, when it is
expressed as function of the re-normalised variable $\theta$. In
particular, the slope of the wings predicted by Corollary~\ref{volimp.cor} 
is achieved rather quickly. 
\item The distinction between the U-shaped smile in the presence of
  a diffusion component and the V-shaped smile in the pure jump case,
  is clearly visible. 
\end{itemize}

\begin{figure}
\centerline{\includegraphics[width = 0.55\textwidth]{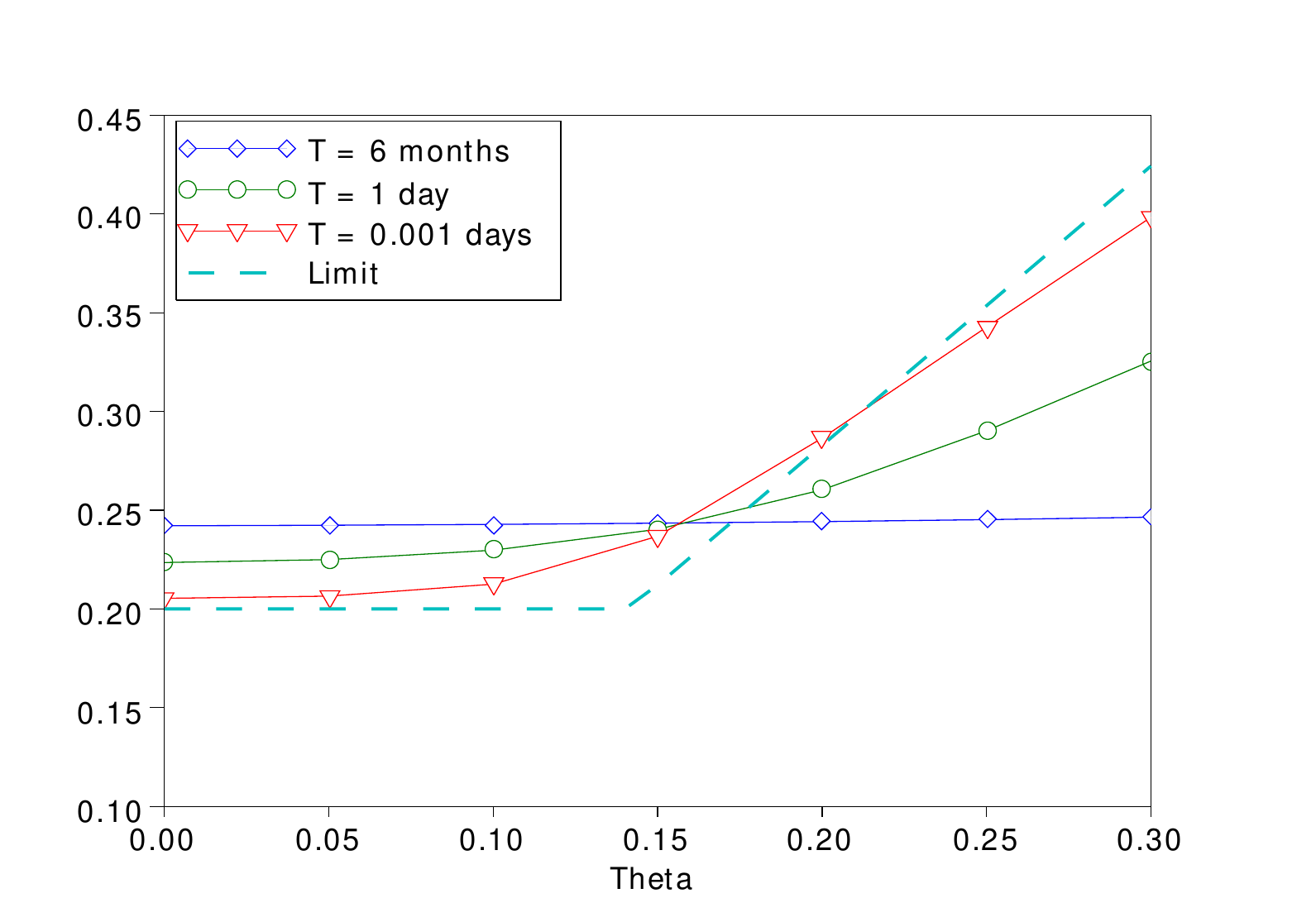}\includegraphics[width = 0.55\textwidth]{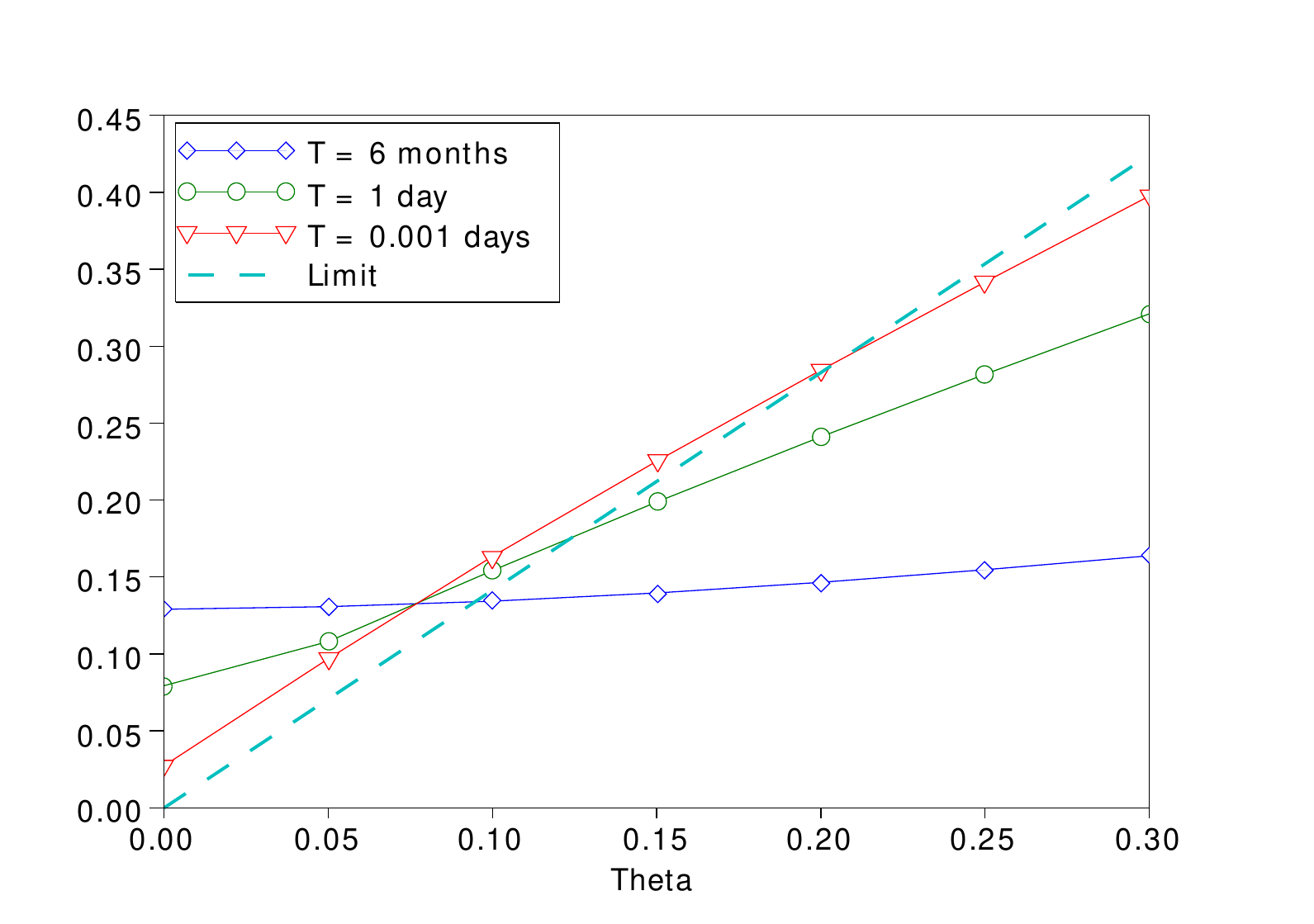}}
\caption{Convergence of the implied volatilities. Left: a diffusion
  component is present. Right: no diffusion component.  }
\label{ivol}
\end{figure}
 
\subsection{Approximation of the implied volatility for small times to
  maturity} In this section we illustrate the approximation of the
implied volatility at small times by the asymptotic formula~\eqref{assform}. We take the same parameters of the tempered stable
process as in Section~\ref{subsec:IVol_Asym} and consider the case
$\sigma=0$ (when the diffusion component is present, in the region
where the pure jump component dominates, the asymptotic formula is the
same, and in the diffusion-dominated region, there are no additional
terms added to the constant limit). Figure~\ref{ivass} illustrates the
quality of the approximation for  $t=1$ day and
  $t=0.1$ days. 
\begin{figure}
\centerline{\includegraphics[width = 0.55\textwidth]{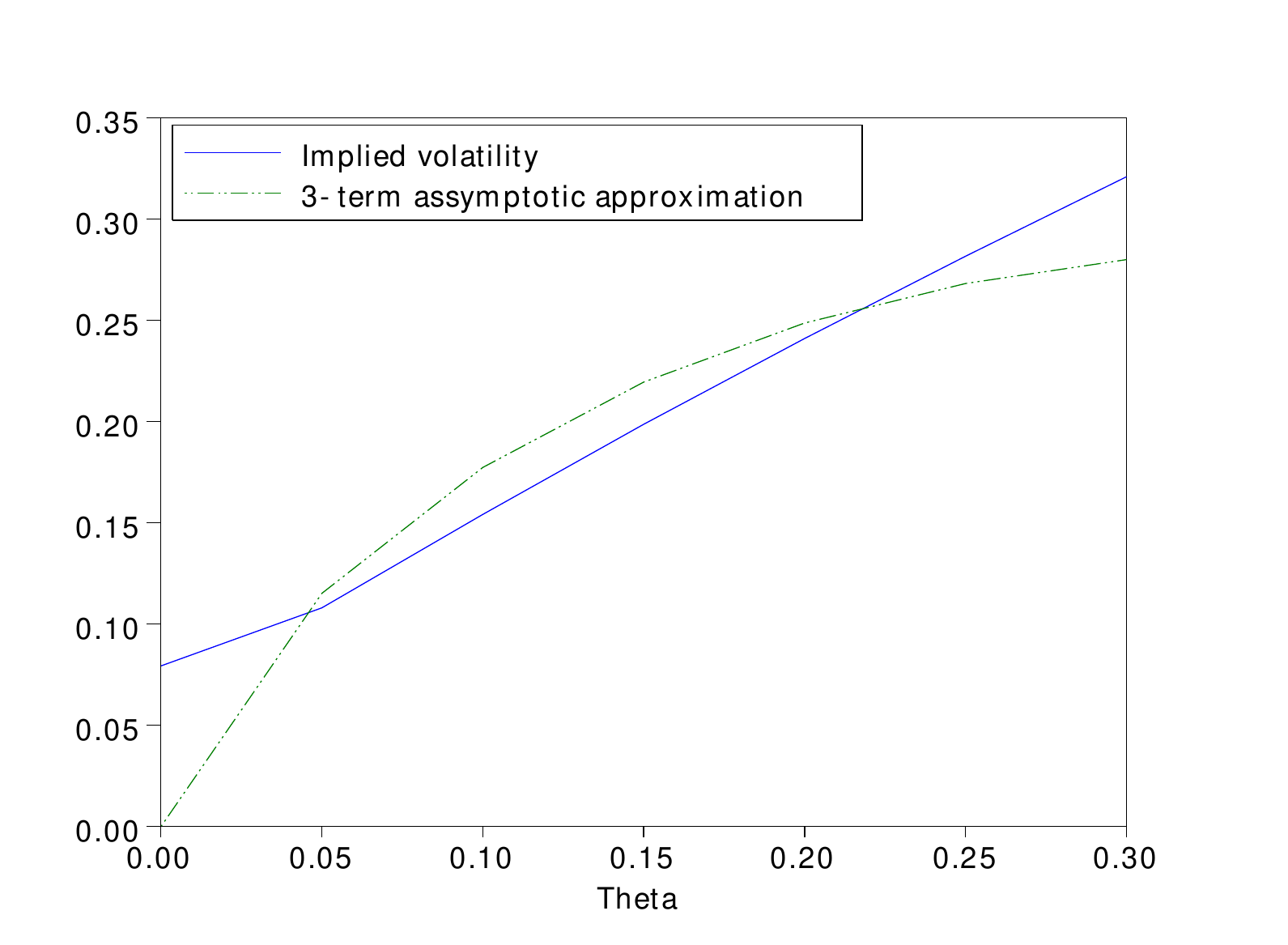}\includegraphics[width = 0.55\textwidth]{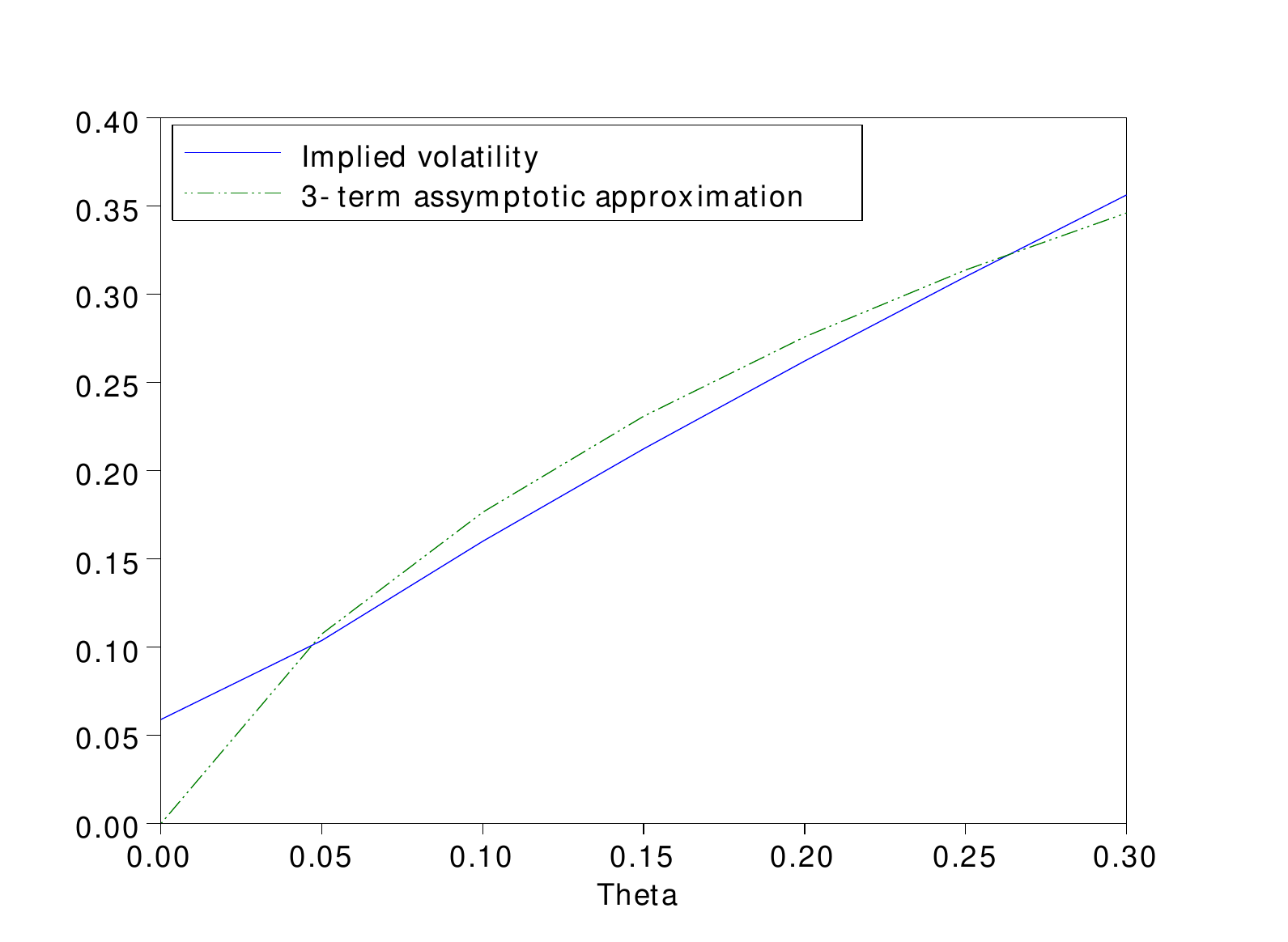}}
\caption{Approximation of the implied volatilities by the asymptotic
  formula \eqref{assform}. Left: $t=1$ day. Right:
  $t=0.1$ days. }
\label{ivass}
\end{figure}

\section{Proofs}
\label{sec:Proofs}

\subsection{Proof of Theorem~\ref{infvar}}
\label{sec:Proof_Thm_infvar}
By Lemma \ref{explin.lm}, to prove Theorem~\ref{infvar}, it is
sufficient to show that 
\begin{align}
\mathbb E[(X_t - k_t)^+] = \mathbb E[(\sigma W_t
- k_t)^+] +  \frac{t k_t^{1-\alpha_+}
  c_+}{\alpha_+-1} + o(tk_t^{1-\alpha_+}+ \mathbb E[(\sigma W_t
- k_t)^+])\label{callcase}
\end{align}
as $t\downarrow 0$ for the call case and 
\begin{align}
\mathbb E[(-k_t - X_t)^+] = \mathbb E[(-k_t - \sigma W_t)^+] +  \frac{t k_t^{1-\alpha_-}
  c_-}{\alpha_--1} + o(tk_t^{1-\alpha_-}+\mathbb E[(-k_t - \sigma W_t)^+]) \label{putcase}
\end{align}
as $t\downarrow 0$ for the put case. Note that~\eqref{putcase} follows from~\eqref{callcase} 
by a substitution $X \mapsto -X$. Therefore, from
now on we concentrate on the proof of~\eqref{callcase}, assuming with
no loss of generality that $c_+>0$. 

\textbf{Step 1.}\quad In this first step, we assume that
$\nu((-\infty,0))=0$ and would like to prove
\begin{align}
\mathbb E[(X_t - k_t)^+] = \mathbb E[(\sigma W_t
- k_t)^+] +  \frac{t k_t^{1-\alpha_+}
  c_+}{\alpha_+-1} + o(tk_t^{1-\alpha_+}).\label{callcase1}
\end{align}
Fix $t>0$, and $\varepsilon>0$ with
$\varepsilon < \frac{1}{32}$, let $X^t$ be a L\'evy process
with no diffusion part, L\'evy measure 
$\nu(dx) 1_{\{0< x\leq\varepsilon k_t\}}$
and third component of the characteristic triplet 
$$
\gamma_t
= \gamma - \int_{(\varepsilon k_t ,1]} z \nu(dz)  
$$
Let $(\xi^t_i)_{i\geq 1}$ be a sequence of i.i.d. random variables
with probability distribution 
$$
\frac{\nu(dz) 1_{\{z> \varepsilon k_t\}}}{\nu(\{z:z> \varepsilon k_t\})},
$$
and 
$N^t$ 
a standard Poisson process with intensity
$\lambda_t:= \nu(\{z:z> \varepsilon k_t\})$. 
Furthermore we assume that 
$X^t$, 
$N^t$
and
$(\xi^t_i)_{i\geq 1}$
are independent.
Then the following equality in law holds
\begin{eqnarray}
\label{eq:Lad_Decomp}
X_t \stackrel{d}{=} \sigma W_t + X^t_t + \sum_{i=1}^{N^t_t} \xi^t_i,
\end{eqnarray}
and it follows that 
\begin{align}
\mathbb E\left[\left(X_t-k_t\right)^+\right] &=
e^{-\lambda_t t} \mathbb E\left[\left(\sigma W_t + X^t_t-k_t\right)^+\right]\label{term1}\\
&+ \lambda_t t  e^{-\lambda_t t} 
\mathbb E\left[\left(\sigma W_t + X^t_t + \xi^t_1-k_t\right)^+\right]\label{term2}\\
& +  e^{-\lambda_t t} \sum_{k\geq 2} \frac{(\lambda_t t)^k}{k!}
\mathbb E\left[\left(\sigma W_t + X^t_t + \sum_{i=1}^k \xi^t_i-k_t\right)^+\right].\label{term3}
\end{align}
As a preliminary computation, we deduce from the assumptions of the theorem that 
the following asymptotic behaviour holds as $t\downarrow 0$
(recall definition~\eqref{def:sim}):
\begin{align}
&\lambda_t \sim c_+(\varepsilon k_t)^{-\alpha_+},\quad 
\Sigma_t:=\int_{(0,k_t \varepsilon]} z^2 \nu(dz) \sim \frac{2c_+}{2-\alpha_+}
(\varepsilon k_t)^{2-\alpha_+},\label{equiv1}\\
& \gamma_t \sim \frac{c_+}{(\varepsilon k_t)^{\alpha_+-1}},\quad \mathbb E[(X_t^t)^2] = t^2 (\gamma_t)^2
+ t \Sigma_t \sim t\Sigma_t,\quad \mathbb E[\xi^t_1]\sim \frac{\alpha_+ }{\alpha_+-1}\varepsilon k_t,\\
& \mathbb E[(X_t^t)^4] = t \int_{(0, k_t \varepsilon]} z^4 \nu(dz) + 4 t^2
\gamma_t \int_{(0,k_t \varepsilon]} z^3 \nu(dz)
+ 3 t^2 \Sigma_t^2
+ 6 t^3 \gamma_t^2 \Sigma_t + t^4 \gamma_t^4 
\nonumber \\
&\qquad \qquad 
\sim\frac{4c_+}{4-\alpha_+}tk_t^{4-\alpha_+}.
 \label{equiv2}
\end{align}
\textbf{To estimate the term in \eqref{term1},} we apply the argument inspired by
Lemma~2 in \cite{ruschendorf02}. 
In the current notation this implies
\begin{align}
\mathbb P[X^t_t > k_t] \leq
\exp\left\{-t\int_{\gamma_t}^{k_t/t} \tau(z)dz\right\},\label{rw}
\end{align}
where $\tau: [\gamma_t,\infty)\to\RR$ is the inverse function
of $s: [0,\infty)\to\RR$ defined by
$$
s(x) = \gamma_t + \int_{(0,k_t \varepsilon]}
z(e^{zx}-1)\nu(dz). 
$$
By Taylor's theorem, this function satisfies
$$
s(x) \leq \gamma_t + x e^{k_t
  \varepsilon x}  \Sigma_t \leq \gamma_t + \frac{e^{2k_t
  \varepsilon x}-1}{k_t \varepsilon}  \Sigma_t.
$$ 
This implies that
$$
\tau(z)\geq \frac{1}{2k_t \varepsilon}\log\left\{1 +
  \frac{z-\gamma_t}{\Sigma_t} k_t \varepsilon\right\},
$$
and therefore, substituting this into \eqref{rw},
\begin{align}
\mathbb P[X^t_t > k_t] &\leq
\exp\left\{-\frac{t\Sigma_t}{2(k_t\varepsilon)^2}\int_{0}^{\frac{k_t
    \varepsilon}{\Sigma_t}(k_t/t - \gamma_t)}
\log(1+s)ds\right\}\notag\\
&\leq \exp\left\{-\frac{k_t - \gamma_t t}{2 k_t\varepsilon}\log\left(\frac{k_t
    \varepsilon}{e \Sigma_t}(k_t/t - \gamma_t)\right)\right\}\label{ruschwor}
\end{align}
From the assumptions of the theorem and \eqref{equiv1}--\eqref{equiv2}, there
exists $t_1>0$ such that $t<t_0$ implies
$$
\mathbb P[X^t_t > k_t] \leq \exp\left\{-\frac{1}{4\varepsilon}\log\left(\frac{k_t^2
    \varepsilon}{2e t \Sigma_t}\right)\right\} = \left(\frac{k_t^2
    \varepsilon}{2e t
    \Sigma_t}\right)^{-\frac{1}{4\varepsilon}} \leq C
(tk_t^{-\alpha_+})^{\frac{1}{4\varepsilon}} \leq C
(tk_t^{-\alpha_+})^{8}
$$ 
for some constant $C<\infty$. By similar arguments it can be shown
that 
\begin{align}
\mathbb P\left[X^t_t > \frac{k_t}{2}\right] \leq C (tk_t^{-\alpha_+})^{4}.\label{estproba}
\end{align}

Coming back to the estimation of \eqref{term1}, we first deal with
the case $\sigma=0$. In this case, the Cauchy-Schwartz inequality
allows to conclude that 
$$
\mathbb E\left[\left(X^t_t-k_t\right)^+\right] \leq
\mathbb E\left[\left({X^t_t}\right)^2\right]^{\frac{1}{2}} \mathbb P[X^{t}_t > k_t]^{\frac{1}{2}} =
O(k_t(tk_t^{-\alpha_+})^2),
$$
because the first factor remains bounded by
\eqref{equiv1}--\eqref{equiv2}. 

Let us now focus on the case $\sigma>0$. Let $f(x) :=
\frac{1}{\sqrt{2\pi}}\int_{x}^\infty (z-x)
e^{-\frac{z^2}{2}}dz$. The expectation in \eqref{term1} can be expressed as
\begin{align*}
\mathbb E[(\sigma W_t + X^t_t - k_t)^+ ]= \sigma \sqrt{t}
\mathbb E\left[f\left(\frac{k_t - X^t_t}{\sigma\sqrt{t}}\right)\right].
\end{align*}
By Taylor's formula, we then get
\begin{align*}
&\mathbb E[(\sigma W_t + X^t_t - k_t)^+ ] = \sigma \sqrt{t}
f\left(\frac{k_t}{\sigma\sqrt{t}}\right) -
f'\left(\frac{k_t}{\sigma\sqrt{t}}\right)\mathbb E[X^t_t] \\ &\qquad \qquad +
\frac{1}{\sigma\sqrt{t}} \mathbb E\left[(X_t^t)^2 \int_0^1 (1-\theta)
  f^{\prime\prime}\left(\frac{k_t - \theta
      X^t_t}{\sigma\sqrt{t}}\right) d\theta\right]\\
& = \mathbb E[(\sigma W_t  - k_t)^+ ] + \gamma_t t \mathbb P[\sigma W_t > k_t] +
\frac{1}{\sigma\sqrt{2\pi t}} \mathbb E\left[(X_t^t)^2 \int_0^1 (1-\theta)  e^{-\frac{1}{2}\left(\frac{k_t - \theta
      X^t_t}{\sigma\sqrt{t}}\right)^2 }d\theta\right].
\end{align*}
We now need to show that the second and the third terms do not
contribute to the limit. Since by assumption $\frac{\sqrt{t}}{k_t} \to
0$, we have that $\mathbb P[\sigma W_t>k_t] \to 0$ as $t \to 0$, and
therefore, by \eqref{equiv1}--\eqref{equiv2}, 
$$
\gamma_t t \mathbb P[\sigma W_t > k_t] = o(tk_t^{1-\alpha_+}). 
$$
The last term can be split into two terms, which are easy to estimate
using \eqref{equiv1}--\eqref{equiv2}:
\begin{align*}
&\frac{1}{\sqrt{t}} \mathbb E\left[(X_t^t)^2 1_{\{X^t_t \leq \frac{k_t}{2}\}}\int_0^1 (1-\theta)  e^{-\frac{1}{2}\left(\frac{k_t - \theta
      X^t_t}{\sigma\sqrt{t}}\right)^2 }d\theta\right] \leq
\frac{1}{\sqrt{t}} \mathbb E[(X_t^t)^2]
e^{-\frac{1}{8}\left(\frac{k_t}{\sigma\sqrt{t}}\right)^2 } \\ &= O(t
k_t^{1-\alpha_+})
\frac{k_t}{\sqrt{t}}e^{-\frac{1}{8}\left(\frac{k_t}{\sigma\sqrt{t}}\right)^2
} = o(t k_t^{1-\alpha_+}),
\end{align*}
because by assumption of the theorem, $\frac{k_t}{\sqrt{t}} \to
\infty$. On the other hand,
\begin{align*}
&\frac{1}{\sqrt{t}} \mathbb E\left[(X_t^t)^2 1_{\{X^t_t > \frac{k_t}{2}\}}\int_0^1 (1-\theta)  e^{-\frac{1}{2}\left(\frac{k_t - \theta
      X^t_t}{\sigma\sqrt{t}}\right)^2 }d\theta\right] \leq
\frac{1}{\sqrt{t}} \mathbb E\left[(X_t^t)^2 1_{\{X^t_t > \frac{k_t}{2}\}}\right] \\
&\leq \frac{1}{\sqrt{t}} \mathbb E[(X_t^t)^4]^{\frac{1}{2}} \mathbb P[X^t_t >
\frac{k_t}{2}]^{\frac{1}{2}} = O(k_t^{2-\frac{\alpha_+}{2}})
O((tk_t^{-\alpha_+})^2) = o(t k_t^{1-\alpha_+})
\end{align*}
by \eqref{equiv2} and \eqref{estproba}. We have therefore shown that 
$$
\mathbb E[(\sigma W_t + X^t_t - k_t)^+ ] = \mathbb E[(\sigma W_t - k_t)^+ ] + o(t k_t^{1-\alpha_+}).
$$
From \eqref{equiv1}, the assumption on $k_t$
in Theorem~\ref{infvar} and the Lipschitz
property of the function $x\mapsto x^+$, it follows that 
$$
e^{-\lambda_t t}\mathbb E[(\sigma W_t + X^t_t - k_t)^+ ] = \mathbb E[(\sigma W_t - k_t)^+ ] + 
o(t k_t^{1-\alpha_+})
$$
as well. 

\textbf{For the term in~\eqref{term2}}, the Lipschitz property of the
function $x\mapsto x^+$, \eqref{equiv1}--\eqref{equiv2}
and the assumption of the theorem (i.e. the first
assumption on 
$k_t$
in Theorem~\ref{infvar}
in the case $\sigma=0$ and the second one
otherwise) 
imply the following estimate:
\begin{align*}
&\lambda_t t\left| \mathbb E\left[\left(\sigma W_t + X^t_t + \xi^t_1-k_t\right)^+\right] -
  \mathbb E\left[\left(\xi^t_1-k_t\right)^+\right]\right|\leq \lambda_t t \{{\mathbb E[|X_t^t|]}  +
\sigma \mathbb E[|W_t|]\}\\& \leq
\lambda_t t\{ \mathbb E[|X_t^t|^2]^{1/2} + \sigma \sqrt{t}\} = O(\lambda_t t^{\frac{3}{2}}
\Sigma_t^{\frac{1}{2}}) + \sigma \lambda_t t^{\frac{3}{2}} 
= o(tk_t^{1-\alpha_+})\quad \text{as $t\to 0$}.
\end{align*}
On the other hand, integration by parts implies
$$
\lambda_t t \mathbb E\left[\left(\xi^t_1-k_t\right)^+\right] = t\int_{k_t}^\infty
\left(z-k_t\right)\nu(dz) = t
\int_{k_t}^\infty U(z)dz\sim
\frac{t k_t^{1-\alpha_+} c_+}{\alpha_+-1}
\quad \text{as $t\to 0$},
$$
where
$U(z):= \nu((z,\infty))$,
which yields the second term in \eqref{infvar.eq}.

\textbf{To treat the summand in \eqref{term3},} observe that by
\eqref{equiv1}--\eqref{equiv2}, for $k\geq 2$,
\begin{align*}
&\mathbb E\left[\left(\sigma W_t + X^t_t + \sum_{i=1}^k
    \xi^t_i-k_t\right)^+\right] \leq  \sigma \sqrt{t} + \mathbb E[|X^t_t|^2]^{1/2}
+ k\mathbb E[\xi^t_1] \\
& = \sigma\sqrt{t} + O(t^{1/2} k_t^{1-\alpha_+/2}) +
k O(k_t) = kO(k_t).
\end{align*}
Therefore, the summand in \eqref{term3} is of order
$O(k_t\lambda^2_t t^2) = O(k_t(tk_t^{-\alpha_+})^2)$
and hence 
$o(tk_t^{1-\alpha_+})$.

\textbf{Step 2.}\quad We now treat the case when
$\nu((-\infty,0))\neq 0$. Let $X^-$ be a spectrally negative L\'evy process with zero
mean and zero diffusion part and $Y$ be a spectrally positive L\'evy
process such that $X^- + Y \stackrel{d}{=} X$. Let
$\beta\in(\max(\alpha_+,\alpha_-),\alpha)$ (where we take $\alpha=2$
is $\sigma>0$) and $\chi_t = t^{\frac{1}{\beta}}$. As before, we fix $\varepsilon>0$ and let $\bar X^t$ be
a L\'evy process with no diffusion part, zero mean and L\'evy measure
$\nu(dx)1_{\{-\varepsilon \chi_t \leq x < 0\}}$, let $\bar \gamma_t =
\int_{(-\infty, - \varepsilon \chi_t)} z\nu(dz)$, let $(\bar
\xi^t_i)_{i\geq 1}$ be a sequence of i.i.d. random variables with
probability distribution 
$$
\frac{\nu(dz)1_{\{z< \varepsilon \chi_t\}}}{\nu(\{z:z<\varepsilon
  \chi_t\})} 
$$
and finally $\bar \lambda_t = \nu(\{z:z<\varepsilon
  \chi_t\})$. With a decomposition similar to
  \eqref{term1}--\eqref{term3}, it is easy to show that the option price $\mathbb E[(X_t-k_t)^+]$ admits an upper bound
\begin{align*}
\mathbb E[(X_t-k_t)^+] &= \mathbb E[(X^-_t + Y_t -k_t)^+] \leq \mathbb E[(\bar X^t_t + \bar
\gamma_t t+ Y_t -k_t)^+] \\
&\leq \mathbb E[( Y_t + \chi_t - k_t)^+] \mathbb P[\bar X^t_t \leq \chi_t] +
\mathbb E[X_t^2]^{\frac{1}{2}} \mathbb P[\bar X^t_t > \chi_t]^{\frac{1}{2}}
\end{align*}
and a lower bound
\begin{align*}
\mathbb E[(X_t-k_t)^+] &= \mathbb E[(X^-_t + Y_t -k_t)^+] \geq  e^{-\bar \lambda_t t} \mathbb E[(\bar X^t_t + \bar
\gamma_t t+ Y_t -k_t)^+]\\
&\geq e^{-\bar \lambda_t t} \mathbb P[\bar X^t_t \geq -\chi_t] \mathbb E[(-\chi_t + \bar
\gamma_t t+ Y_t -k_t)^+]
\end{align*}
Similarly to \eqref{equiv1}--\eqref{equiv2}, we have
$$
\bar \Sigma_t := \int_{(-\varepsilon \chi_t,0)} z^2 \nu(dz) \sim
\frac{2}{2-\alpha_-} (\varepsilon \chi_t)^{2-\alpha_-}, 
$$
and with the same logic as in \eqref{ruschwor}, we have that
$$
\mathbb P[\bar X^t_t > \chi_t] \leq \left(\frac{\chi_t^2 \varepsilon}{e
    \bar\Sigma_t t }\right)^{\frac{1}{2\varepsilon}} \sim
\left(t^{\frac{\alpha_-}{\beta} -1}\right)
^{\frac{1}{2\varepsilon}},\quad t\to 0.
$$
It is now clear that one can choose $\varepsilon>0$ so that the square
root of this
expression becomes equal to $o(t k_t^{1-\alpha_+})$. Since $\mathbb P[\bar
X^t_t < -\chi_t]$ admits the same estimate, and $t\bar \lambda_t \to
0$ as $t\to 0$, we get that
\begin{align*}
\mathbb E[(X_t-k_t)^+] &\geq m_t \mathbb E[(Y_t-\chi_t + \bar
\gamma_t t -k_t)^+]\\
\mathbb E[(X_t-k_t)^+] &\leq M_t \mathbb E[(Y_t+\chi_t  -k_t)^+] + o(t k_t^{1-\alpha_+}),
\end{align*}
where $m_t$ and $M_t$ converge to $1$ as $t\to 0$. 
Since $\chi_t = o(k_t)$ and $\bar \gamma_t t = o(k_t)$, from \eqref{callcase1}, we then get
\begin{align*}
\mathbb E[(X_t-k_t)^+] &\geq m_t \mathbb E[(\sigma W_t-\chi_t + \bar
\gamma_t t -k_t)^+] + m_t\frac{t(k_t + \chi_t - \bar \gamma_t
  t)^{1-\alpha_+} c_+}{\alpha_+-1}+ o(t k_t^{1-\alpha_+})\\
& = m_t \mathbb E[(\sigma W_t-\chi_t + \bar
\gamma_t t -k_t)^+] + \frac{t k_t^{1-\alpha_+} c_+}{\alpha_+-1}+ o(t k_t^{1-\alpha_+})\\
\mathbb E[(X_t-k_t)^+] &\leq M_t \mathbb E[(\sigma W_t+\chi_t  -k_t)^+] + M_t\frac{t(k_t
  - \chi_t )^{1-\alpha_+} c_+}{\alpha_+-1}+ o(t k_t^{1-\alpha_+})\\
 & = M_t \mathbb E[(\sigma W_t+\chi_t  -k_t)^+] + \frac{t k_t^{1-\alpha_+} c_+}{\alpha_+-1}+ o(t k_t^{1-\alpha_+})
\end{align*}
Finally, since we also have $\chi_t = o(\sqrt{t})$ and $\bar \gamma_t
t = o(\sqrt{t})$, we get that $\mathbb E[(\sigma W_t-\chi_t + \bar
\gamma_t t -k_t)^+] \sim \mathbb E[(\sigma W_t -k_t)^+]$ and $\mathbb E[(\sigma
W_t+\chi_t -k_t)^+] \sim \mathbb E[(\sigma W_t -k_t)^+]$, which allows to
complete the proof 
of Theorem~\ref{infvar}.

\subsection{Proof of Proposition~\ref{finvar}}
\label{sec:Proof_Prop_finvar}

We first concentrate on the proof of~\eqref{finvar2.eq}. Let
$(\sigma^2,\nu,b)$ be the characteristic triplet of $X$ with
respect to zero truncation function, meaning that 
$$
X_t = bt + \sigma W_t + \sum_{s\leq t} \Delta X_s, 
$$
where as usual
for any
$s>0$
we define
$\Delta X_s=X_s-X_{s-}$.

Assume first that $\sigma=0$. The left-derivative of the function 
$$
x\mapsto(e^{-k_t}-e^x)^+\qquad\text{is}\qquad
x\mapsto -e^{x}1_{\{x\leq-k_t\}},
$$
and hence It\^o-Tanaka formula~\cite[Ch.~IV,~Thm.~70]{protter2nd} applied to 
the process
$(e^{-k_t}-e^X)^+$
yields
\begin{eqnarray*}
(e^{-k_t}-e^{X_t})^+ & = & -\int_{(0,t]}e^{X_{s-}}1_{\{X_{s-}\leq-k_t\}}dX_s\\
& + &
\sum_{0<s\leq t} \left[(e^{-k_t} -e^{X_s} )^+ - (e^{-k_t} -e^{X_{s-}} )^++ e^{X_{s-}}1_{\{X_{s-}\leq-k_t\}}\Delta X_s\right] \\ 
& = & -b \int_0^t e^{X_{s-}}1_{\{X_{s-}\leq-k_t\}}ds + 
\sum_{0<s\leq t}\left[ (e^{-k_t} -e^{X_{s-}+\Delta X_s} )^+ - (e^{-k_t} -e^{X_{s-}} )^+\right]
\end{eqnarray*}
for any
$t\geq0$,
since, in this case,
$X$
has paths of finite variation.
Since 
$(\Delta X_s)_{s\geq0}$
is a Poisson point process with intensity measure 
$\nu(dy)\times ds$,
and 
$X_{s-}\neq X_s$
for at most countably many time
$s$
in the interval
$[0,t]$
almost surely,
taking expectations on both sides of the path-wise representation above
and applying the compensation formula for point processes
yields
\begin{eqnarray}
\label{ito1}
\mathbb E[(e^{-k_t} - e^{X_t})^+] 
&= &-b \mathbb E\left[\int_{0}^t e^{X_{s}} 1_{\{X_{s} \leq -k_t\}} 
  ds\right] \\ 
& + & \mathbb E\left[\int_0^t \int_{\mathbb R\setminus\{0\}} \left\{(e^{-k_t} -e^{X_s
      + y})^+ - (e^{-k_t} -e^{X_{s}} )^+\right\} \nu(dy)ds\right].\nonumber
\end{eqnarray}
From Theorem 43.20 in~\cite{sato}, we have that $\frac{X_t}{t} \to b$ almost surely
as $t\to 0$. Therefore, for any 
$\varepsilon>0$,
each path 
$X(\omega)$
satisfies the following inequalities 
$$
X_t(\omega)>t(b-\varepsilon)>-k_t\qquad\text{for all small enough}\qquad t>0
$$
(recall that by assumption
$k_t/t\to\infty$
as
$t\downarrow0$).
Furthermore, 
since
$k_t\downarrow0$
as
$t\downarrow0$,
for all sufficiently small $t$ 
we have
$X_s(\omega) > -k_t$ for all $s\leq t$. 
Therefore it holds 
$\frac{1}{t}\int_0^t e^{X_s} 1_{\{X_s \leq -k_t\}} ds \to 0$ 
almost surely. Since
on the other hand we have
$$
\frac{1}{t}\int_0^t e^{X_s} 1_{\{X_s \leq -k_t\}} ds \leq
\frac{1}{t}\int_0^t e^{-k_t} ds = e^{-k_t},
$$ 
the dominated convergence theorem
implies
$$
\mathbb E\left[\int_{0}^t e^{X_{s}} 1_{\{X_{s} \leq -k_t\}} 
  ds\right] = o(t)\quad \text{as $t\to 0$}. 
$$

To deal with the second term in \eqref{ito1}, observe that for any
$\varepsilon>0$, 
each path 
$X(\omega)$
satisfies the inequalities 
$(b-\varepsilon )t
\leq X_t(\omega) \leq (b+\varepsilon)t$
for all $t$ sufficiently small. 
Therefore
$X(\omega)$
also satisfies the following
inequalities for any
$y\in\RR\setminus\{0\}$
and all sufficiently small times $t>0$:
$$
(e^{-k_t} -e^{X_s(\omega) + y})^+ - (e^{-k_t} -e^{X_{s}(\omega)} )^+ 
\leq (e^{-k_t} -e^{(b-\varepsilon)t + y})^+ - (e^{-k_t} -e^{(b-\varepsilon)t} )^+
$$
and 
$$
(e^{-k_t} -e^{X_s(\omega) + y})^+ - (e^{-k_t} -e^{X_{s}(\omega)} )^+ 
\geq (e^{-k_t} -e^{(b+\varepsilon)t + y})^+ - (e^{-k_t} -e^{(b+\varepsilon)t} )^+.
$$
The second term in both sides of the above inequalities is in fact
always zero for sufficiently small $t$. 
Therefore we get the following almost sure convergence:
$$
\frac{1}{t}\int_0^t \int_{\mathbb R\setminus\{0\}} \left\{(e^{-k_t} -e^{X_s + y})^+ -
    (e^{-k_t} -e^{X_{s}} )^+\right\} \nu(dy)ds \to \int_{\mathbb R\setminus\{0\}}(1-e^{y})^+\nu(dy)
    \qquad\text{as $t\downarrow0$.}
$$
Since the function $y\mapsto (e^{-k_t} - e^{X_s + y})^+$ is
Lipschitz with a Lipschitz constant that does not depend
on the path
$X(\omega)$,
the dominated convergence theorem 
and the representation in~\eqref{ito1}
yield
$$
\lim_{t\downarrow 0} \frac{1}{t} \mathbb E[(e^{-k_t} - e^{X_t})^+] = \int_{\mathbb R\setminus\{0\}}(1-e^{y})^+\nu(dy).
$$

Assume now that $\sigma>0$. Define 
$$
f(t,x) := \mathbb E\left[\left(1-e^{x + k_t + \sigma W_t - \frac{\sigma^2}{2} t}\right)^+\right]\qquad\text{and}\qquad
Z_t:= \left(b +\frac{\sigma^2}{2}\right)t + \sum_{0<s\leq t} \Delta X_s
$$
and note that
\begin{equation}
\label{eq:rep_Z_exp}
\mathbb E[(e^{-k_t} - e^{X_t})^+] = e^{-k_t} \mathbb E[f(t,Z_t)]. 
\end{equation}
The derivative of $f$ with respect to $x$ is given by
\begin{equation}
\label{eq:Lower_Bound}
f'(t,x) := -\mathbb E\left[e^{x + k_t + \sigma W_t - \frac{\sigma^2}{2}
    t}1_{\{x + k_t + \sigma W_t - \frac{\sigma^2}{2} t\leq 0\}}\right]\geq -1,
\end{equation}
It can be
computed explicitly as
\begin{equation}
\label{eq:Explicit_Der}
f'(t,x) = - e^{x+k_t} N\left(-\frac{\frac{1}{2} \sigma^2 t + x+k_t}{\sigma\sqrt{t}}\right),
\end{equation}
where $N(\cdot)$ denotes the standard normal CDF. Note also for future use
that
\begin{align}
f^{\prime\prime}(t,x) = - e^{x+k_t} N\left(-\frac{\frac{1}{2} \sigma^2
    t + x+k_t}{\sigma\sqrt{t}}\right) + \frac{1}{\sigma\sqrt{t}} n\left(\frac{-\frac{1}{2} \sigma^2
    t + x+k_t}{\sigma\sqrt{t}}\right)\geq -1,\label{2der}
\end{align}
with $n(x) = N'(x)$.

Applying It\^o's formula to the process
$f(t,Z)$
as a function of $Z$ 
with
$t$
fixed, yields
\begin{eqnarray*}
f(t,Z_t) & = & f(t,0)+\int_{(0,t]}f'(t,Z_{s-})dZ_s+\sum_{0<s\leq t}\left[f(t,Z_s)-f(t,Z_{s-})-f'(t,Z_{s-})\Delta Z_s\right]\\
& = & f(t,0)+\left(b +\frac{\sigma^2}{2}\right) \int_{0}^tf'(t,Z_{s-})ds+\sum_{0<s\leq t}\left[f(t,Z_{s-}+\Delta X_s)-f(t,Z_{s-})\right],
\end{eqnarray*}
since
$\Delta Z_s=\Delta X_s$
for all 
$s>0$.
By taking the expectation and applying~\eqref{eq:rep_Z_exp} we find
\begin{align}
\label{ito2}
\mathbb E[(e^{-k_t} - e^{X_t})^+] &= 
e^{-k_t} \mathbb E[f(t,0)] +
e^{-k_t}\mathbb E\left[ \int_0^t f'(t,Z_{s}) 
ds \right] \\
&\qquad \qquad + e^{-k_t}\mathbb
E\left[\int_0^t \int_{\mathbb R\setminus\{0\}} \left\{f(t,Z_s + y)-f(t,Z_s)\right\}\nu(dy) ds\right].
\notag
\end{align}

The first term on the
right-hand side of~\eqref{ito2}
is equal to the first term on the right-hand side of~\eqref{finvar2.eq}. 
As in the case $\sigma=0$, using the almost sure convergence 
$\frac{Z_t}{t} \to b + \frac{\sigma^2}{2}$, 
the explicit form~\eqref{eq:Explicit_Der} of $f'(t,x)$ 
and the assumption that
$\frac{k_t}{\sqrt{t}}\to\infty$ as $t \downarrow 0$,
we get that 
$$
\frac{1}{t} \int_0^t f'(t,Z_s) ds \to 0
$$
almost surely. 
Since
$|f'(t,Z_s)|\leq 1$
for all
$t,s\geq0$
by~\eqref{eq:Lower_Bound},
the dominated convergence theorem yields
$$
\mathbb E\left[ \int_0^t f'(t,Z_{s})ds \right] = o(t). 
$$

To treat the last term in~\eqref{ito2}, we use the fact that for any $\varepsilon>0$,
each path
$Z(\omega)$
satisfies the inequalities
$$
t(b + \sigma^2/2-\varepsilon)\leq Z_t(\omega) \leq t(b + \sigma^2/2+\varepsilon).
$$
for all sufficiently small $t$.
Therefore, 
since $f^{''}(t,x) \geq -1$, 
the following inequalities hold
\begin{equation}
\label{eq:path_bounds}
f'(t, t(b-\varepsilon + \sigma^2/2) + \theta y) - 2t \varepsilon \leq f'(t,Z_s(\omega) + \theta y) \leq 
f'(t, t(b+\varepsilon + \sigma^2/2) + \theta y) + 2t \varepsilon
\end{equation}
for any trajectory
$s\mapsto Z_s(\omega)$,
where
$s\in[0,t]$,
and all sufficiently small $t$.
The random variable under the expectation 
in the last term on the right-hand side of~\eqref{ito2}
can be expressed as follows:
\begin{align}
\frac{1}{t}\int_0^t \int_{\mathbb R\setminus\{0\}} \left\{f(t,Z_s + y)-f(t,Z_s)\right\}\nu(dy)  
ds = \frac{1}{t}\int_0^t ds \int_0^1 d\theta \int_{\mathbb R\setminus\{0\}} y f'(t,Z_s + \theta y)\nu(dy). \label{underexp}
\end{align}
The path-wise bounds in~\eqref{eq:path_bounds}
can be used to estimate~\eqref{underexp} from above and
below. For each path
$Z(\omega)$
we have the following bound for 
$y\in(-\infty,0)$
and
all sufficiently small
$t$:
\begin{align*}
&2t\varepsilon \int_{(-\infty,0)} y \nu(dy) + \int_0^1 d\theta
\int_{(-\infty,0)} y f'(t, t(b+\varepsilon + \sigma^2/2 ) + \theta y) \nu(dy)\\
&\leq \frac{1}{t}\int_0^t ds \int_0^1 d\theta \int_{(-\infty,0)} y f'(t,Z_s(\omega) +
\theta y)\nu(dy)\\ 
&\leq - 2t\varepsilon \int_{(-\infty,0)} y \nu(dy) + \int_0^1 d\theta
\int_{(-\infty,0)} y f'(t, t(b-\varepsilon + \sigma^2/2 ) + \theta y) \nu(dy).
\end{align*}
The explicit form~\eqref{eq:Explicit_Der} of $f'(t,x)$ 
implies 
that for all $y<0$ and $\theta>0$
we have
$$
f'(t, t(b\pm\varepsilon + \sigma^2/2 )
+ \theta y) \to -e^{\theta y}
\qquad\text{as $t\to 0$.}
$$ 
Since $f'(t,x)$ is bounded, the dominated convergence theorem yields
$$
\int_0^1 d\theta
\int_{(-\infty,0)} y f'(t, t(b\pm\varepsilon + \sigma^2/2)
+ \theta y) \nu(dy) \to -\int_0^1 d\theta \int_{(-\infty,0)} y e^{\theta y}
\nu(dy) = \int_{(-\infty,0)} \nu(dy)(1-e^y) 
$$
as
$t\downarrow0$.
Formula~\eqref{eq:Explicit_Der} for $f'(t,x)$ 
implies 
that for all $y\in(0,\infty)$ and $\theta>0$
we have
$$
f'(t, t(b\pm\varepsilon + \sigma^2/2 )
+ \theta y) \to 0 
\qquad\text{as $t\to 0$.}
$$ 
An analogous argument 
for 
$y\in(0,\infty)$
to the one above
and the representation in~\eqref{underexp}
imply the almost sure convergence
$$
\frac{1}{t}\int_0^t \int_{\mathbb R\setminus\{0\}} \left\{f(t,Z_s + y)-f(t,Z_s)\right\}\nu(dy)  
ds \to \int_{\mathbb R\setminus\{0\}}(1-e^y)^+ \nu(dy)
\qquad\text{as $t\to0$.}
$$
Finally, since $f(t,x)$ is Lipschitz in 
$x$,
with the Lipschitz constant independent 
of
$t$,
the dominated convergence theorem implies
$$
\frac{1}{t}\mathbb
E\left[\int_0^t \int_{\mathbb R\setminus\{0\}} \left\{f(t,Z_s +
    y)-f(t,Z_s)\right\}\nu(dy) ds\right] \to \int_{\mathbb R\setminus\{0\}}(1-e^y)^+ \nu(dy).
$$
This concludes the proof of~\eqref{finvar2.eq}. 
Note that in this proof, we did not use the condition 
in~\eqref{eq:Assum_tail},
but only the assumption 
$\int_{\mathbb R\setminus\{0\}} |x|\nu(dx) <\infty$. 

We now concentrate on the proof of~\eqref{finvar1.eq}. 
Since the L\'evy process 
$X$ 
satisfies~\eqref{eq:Assum_tail},
we can define the share measure
$\wt \PP$,
via
$ \frac{d\wt{\mathbb P}}{d\mathbb P}|_{\mathcal F_t} =
e^{X_t}$,
as in the proof of Theorem~\ref{ivexpand}.
Analogous to the equality in~\eqref{eq:put_call_sym},
we have
\begin{equation}
\label{eq:share_measure_identity}
\mathbb E[(e^{X_t}- e^{k_t})^+] = e^{k_t} 
\wt{\mathbb E}[(e^{-k_t} - e^{-X_t})^+], 
\end{equation}
where $\wt{ \mathbb E}$ denotes the expectation under the
share measure
$\wt{\mathbb P}$. 
Furthermore,
it is well-known that under the measure
$\wt \PP$, the process 
$X$ is again a L\'evy process with a characteristic triplet
$(\sigma^2,\wt \nu, \wt \gamma)$, where $\wt \nu(dx) = e^x \nu(dx)$, 
and 
$e^{-X}$
is a positive 
$\wt\PP$-martingale
started at one.
The L\'evy measure $\wt \nu$
clearly satisfies 
$$
\int_{\mathbb R\setminus\{0\}} |x| \wt \nu(dx)< \infty.
$$
Therefore we can apply~\eqref{finvar2.eq} to the process $-X$
under the measure $\wt{\mathbb P}$. 
Hence the identity in~\eqref{eq:share_measure_identity}
yields:
\begin{align*}
\mathbb E[(e^{X_t}- e^{k_t})^+] & = e^{k_t} 
\mathbb E[(e^{-k_t} - e^{\sigma W_t- \frac{\sigma^2 t}{2}})^+] +
t e^{k_t}    \int_{(0,\infty)}
(1-e^{-x}) \wt\nu(dx) +o(t)\\
& = 
\mathbb E[(e^{\sigma W_t- \frac{\sigma^2 t}{2}}-e^{k_t} )^+] + t  \int_{(0,\infty)} (e^x-1)\nu(dx) +o(t),
\end{align*}
where we used the Black-Scholes put-call symmetry
given in~\eqref{eq:Put_Call_Sy_BS},
the fact $e^{k_t} = 1 + o(1)$ 
and the equality
$\wt \nu(dx) = e^x \nu(dx)$. 
This establishes the formula in~\eqref{finvar1.eq} and concludes the proof of 
Proposition~\ref{finvar}.

\section*{Appendix}

\begin{lemma}\label{explin.lm}
Let $X$ be a L\'evy process 
satisfying~\eqref{eq:Assum_tail}
and
$k_t$ a deterministic function such that 
$$
k_t>0\quad\forall t>0 \quad \text{and}\quad \lim_{t\downarrow} k_t = 0\quad \text{as $t \downarrow 0$}.
$$
Then for any $b\in \mathbb R$
we have
\begin{align*}
\mathbb E[(e^{X_t + bt} - e^{k_t})^+] &= e^{k_t}\mathbb E[(X_t-k_t)^+] + O(t)\quad \text{as
  $t\downarrow 0$,}\\
\mathbb E[(e^{-k_t}-e^{X_t+ bt} )^+] &= e^{-k_t}\mathbb E[(-k_t-X_t)^+] + O(t)\quad \text{as
  $t\downarrow 0$.}
\end{align*}
\end{lemma}
\begin{proof} 
Since 
$0\leq(X_t+bt-k_t)^+-(X_t-k_t)^+\leq b^+ t=O(t)$,
it is clearly sufficient to prove the formula for the call in the case
$b=0$. 
Let $f(x,k) = (e^x-e^k)^+ - e^k (x-k)^+$ and note the following: 
$f'_x(x,k) = (e^x-e^k)^+$ 
for all
$x\in\RR$
and 
$f^{''}_x(x,k) = e^x 1_{\{x\geq k\}}$
for all
$x\in\RR\setminus\{k\}$. 
By Taylor's formula we have 
$f(x,k)=(x-k)^2\int_0^1(1-\theta)f_x''((1-\theta)k+\theta x)d\theta$
for any
$x\neq k$,
and, considering $k_t$ fixed,
we find
\begin{equation*}
\mathbb E[f(X_t,k_t)] 
= \mathbb E\left[(X_t-k_t)^2 \int_0^1 (1-\theta) e^{k_t+\theta (X_t-k_t)} 1_{\{k_t+\theta (X_t-k_t) \geq k_t\}} d\theta \right] \leq C_0 
\mathbb E[X_t^2 e^{X_t}]
\end{equation*}
for some constant
$C_0>0$.
Under the assumption of the lemma, the right-hand side can be computed
as
$$
\mathbb E[X_t^2 e^{X_t}]  = \frac{\partial^2 }{\partial u^2}
\mathbb E[e^{uX_t}]\Big|_{u=1}. 
$$
A direct computation using the L\'evy-Khintchine formula then shows that
$\mathbb E[X_t^2 e^{X_t}] = O(t)$ as $t\downarrow 0$. The put case is treated
in a similar manner. 
\end{proof}


\end{document}